\documentclass{article}

\usepackage[utf8]{inputenc}
\usepackage[T1]{fontenc}
\usepackage[margin=1.1in]{geometry}
\usepackage{amsmath,amssymb,amsthm,mathtools}
\usepackage{booktabs,array,tabularx}
\usepackage{enumitem}
\usepackage{algorithm}
\usepackage[noend]{algpseudocode}
\usepackage{tikz}
\usetikzlibrary{arrows.meta,positioning,calc,decorations.pathreplacing,decorations.pathmorphing,shapes.geometric}
\usepackage{placeins}
\definecolor{wA}{HTML}{7F77DD}\definecolor{wB}{HTML}{1D9E75}\definecolor{wC}{HTML}{D85A30}
\definecolor{wD}{HTML}{D4537E}\definecolor{wE}{HTML}{378ADD}\definecolor{wF}{HTML}{639922}
\definecolor{wG}{HTML}{BA7517}\definecolor{wH}{HTML}{E24B4A}\definecolor{wI}{HTML}{534AB7}
\definecolor{wJ}{HTML}{0F6E56}\definecolor{wK}{HTML}{993C1D}\definecolor{wL}{HTML}{993556}
\tikzset{cubeedge/.style={black!30,line width=0.5pt},
  fixpt/.style={fill=white,draw=black!55,line width=0.6pt},
  ring/.style={draw=black,line width=0.7pt},
  lab/.style={font=\scriptsize\ttfamily,text=black!60}}
\usepackage[hidelinks]{hyperref}
\usepackage[capitalise,nameinlink]{cleveref}
\numberwithin{equation}{section}
\theoremstyle{plain}
\newtheorem{theorem}{Theorem}[section]
\newtheorem{proposition}[theorem]{Proposition}
\newtheorem{lemma}[theorem]{Lemma}
\newtheorem{corollary}[theorem]{Corollary}

\newtheorem{question}[theorem]{Question}
\theoremstyle{definition}
\newtheorem{definition}[theorem]{Definition}
\newtheorem{example}[theorem]{Example}

\theoremstyle{remark}
\newtheorem*{remark}{Remark}
\crefname{nremark}{Remark}{Remarks}
\crefname{question}{Question}{Questions}

\newcommand{\F}{\mathbb{F}}
\newcommand{\A}{\mathcal{A}}
\newcommand{\cP}{\mathcal{P}}
\newcommand{\Sym}{\mathrm{Sym}}
\newcommand{\stab}{\mathrm{stab}}
\newcommand{\orb}{\mathrm{orb}}
\newcommand{\Par}{\mathrm{Par}}
\newcommand{\GL}{\mathrm{GL}}
\newcommand{\CT}{\mathit{CT}}
\newcommand{\soph}{\mathrm{soph}}
\newcommand{\hGL}{h^{\GL}}
\newcommand{\haff}{h^{\mathrm{aff}}}
\newcommand{\hstr}{h^{\mathrm{str}}}
\newcommand{\AGL}{\mathrm{AGL}}
\newcommand{\Alt}{\mathrm{Alt}}
\newcommand{\Sp}{\mathrm{Sp}}
\newcommand{\Stab}{\mathrm{stab}}
\newcommand{\Fix}{\mathrm{Fix}}

\newcommand{\norb}{\#\mathrm{orb}}
\newcommand{\CL}[2]{\mathcal L_{#1}(#2)}

\newcommand{\defi}{\mathrm{def}}
\newcommand{\mv}[1]{\textup{\textsc{#1}}}
\newcommand{\R}{\mathcal{R}}
\DeclareMathOperator{\rank}{rank}

\DeclareMathOperator*{\argmax}{arg\,max}

\title{Orbits as Strong Models:\\[4pt]
\large Kolmogorov Structure Function under Symmetry}
\hypersetup{pdftitle={Orbits as Strong Models: Kolmogorov Structure Function under Symmetry},pdfauthor={Romie Banerjee}}
\author{Romie Banerjee}
\date{November 2026}

\begin{document}
\maketitle

\begin{abstract}
Strong models in algorithmic statistics are essentially the cells of
simple partitions of $\{0,1\}^n$. We study the partitions that come from
symmetry, the orbit partitions of groups acting on strings, and the
structure function $h^G_x$ they induce. A set is a cell of such a
partition exactly when its setwise stabiliser acts transitively on it, so
$h^G_x$ is Kolmogorov's structure function restricted to $G$-homogeneous
sets. For the full symmetric group this recovers all models and exactly
the strong ones. For $\GL(n,2)$ the cells are the linearly homogeneous
sets, $h^{\GL}_x$ lies in a band between $C(x)-\alpha$ and $n-\alpha$
with both edges attained, and some stochastic normal strings have simple
structure that is invisible to linear symmetry. Writing each permutation
group as a Burnside-ring type with a placement, we show that any space of
symmetry hypotheses small enough to search is small enough to miss simple
structure. Finally we cast the search for orbit models as value
estimation on subgroup lattices, with computable bounds, certified
training targets and a lattice version of Monte Carlo tree search.
\end{abstract}

\tableofcontents

\section{Introduction}

\subsection{Models, strong models and partitions}

Kolmogorov proposed to explain a string $x$ by a finite set $S\ni x$, a
\emph{model}, and to measure the explanation by the two-part code
$C(x)\lesssim C(S)+\log|S|$: first the model, then the position of $x$
in it, which is treated as noise~\cite{K74,GTV,VV04,VS17}. The structure
function $h_x(\alpha)$ records the smallest model affordable at
complexity $\alpha$, and a model is a sufficient statistic when the two
parts together cost no more than $C(x)$.

Arbitrary finite sets are too generous as models. The universal models
of G\'acs, Tromp and Vit\'anyi~\cite{GTV} can be sufficient, yet they
encode the number of strings of bounded complexity, so they cannot be
found from the data. Vereshchagin~\cite{V15} therefore singled out the
\emph{strong} models: those computable from $x$ by a total algorithm.
A total algorithm $x'\mapsto A(x')$ assigns a model to \emph{every}
string, and Milovanov~\cite{M16} observed that this makes strong models
essentially the \emph{cells of simple partitions} of $\{0,1\}^n$. So
beneath algorithmic statistics lies a theory of partitions. A partition
$p$ is a \emph{hypothesis}: a classification of all strings into classes
of strings the hypothesis cannot tell apart. The model for $x$ is its
cell, and the description becomes
\begin{equation}\label{eq:threepart}
  C(x)\ \lesssim\ \underbrace{C(p)}_{\text{hypothesis}}
  \ +\ \underbrace{C(B\mid p)}_{\text{which cell}}
  \ +\ \underbrace{\log|B|}_{\text{noise}},
  \qquad x\in B\in p .
\end{equation}

\subsection{Symmetric partitions}

Which partitions are \emph{structured} hypotheses? The most natural
source of classifications in mathematics and physics is symmetry: two
strings are indistinguishable if some transformation from a group maps
one to the other. The hypothesis is then the \emph{orbit partition} of a
group $U$ acting on strings, and the size of the cell of $x$ is governed
by its symmetry. By the orbit--stabiliser theorem,
\[
  \log|U\cdot x|=\log|U|-\log|\stab_U(x)|=-\log\Pr_{u\in U}[ux=x],
\]
so strings with large stabilisers lie in small cells and receive short
descriptions in~\eqref{eq:threepart}. This is the precise sense in which
symmetric objects are simple: a symmetry hypothesis explains $x$ well
when it is cheap and $x$ is fixed by a large part of it.

Different groups can have the same orbits and then make the same
hypothesis, so the group is best seen as a \emph{certificate} that a
partition arises from symmetry. The classical Galois connection between
subgroups and partitions~\cite{Ke99} makes this precise. For an ambient
group $G$ acting on strings it singles out a lattice $\Par_G$ of
\emph{symmetric partitions}. Each symmetric partition has a canonical
certificate, the largest subgroup with those orbits, which costs exactly
as much as the partition. Hypotheses can be combined by joins and meets,
and conjugation acts as relabelling. The theory of this paper is
algorithmic statistics restricted to cells of symmetric partitions. These
cells have an intrinsic description: a set is a cell exactly when its
setwise stabiliser acts transitively on it, and the cell then carries a
canonical hypothesis at no extra cost (Proposition~\ref{prop:cells}). So
the theory is that of a restricted model class in the sense of
Vereshchagin and Vit\'anyi~\cite{VV10}.

The choice of ambient group matters. For $G=\Sym(\{0,1\}^n)$ every
partition is symmetric, and symmetry adds nothing to the theory of
strong models. Its canonical certificates are Young subgroups, whose
type records only the sizes of the cells, so for complex models the
information sits in how the symmetry is placed
(Proposition~\ref{prop:placement}). At the
other extreme, the orbit partition of the stabiliser of $x$ in a large
group has $\{x\}$ as a cell: a lossless description that costs $C(x)$. The informative cases
lie in between and arise from structured ambient groups, above all the
group $\GL(n,2)$ of linear symmetries of $\F_2^n$.

\subsection{The question and the results}

$C(x)$ can always be recovered from a symmetric partition at full price,
so the question is not \emph{whether} but \emph{how}: \emph{which part of
the regularity of $x$ is captured by cells of symmetric partitions, and
at what cost?} We measure this by the structure function $h^G_x$ over
$\Par_G$, and summarise it by the symmetric sophistication $\soph^G(x)$,
the least cost at which a cell of a symmetric partition is a sufficient
statistic.

\begin{enumerate}[label=(\arabic*)]
\item \emph{Framework} (Sections~\ref{sec:galois}--\ref{sec:partsf}).
Symmetric partitions, their canonical certificates and costs, cells as
homogeneous sets, joins, meets and relabelling, and a dictionary between
the group and partition formulations.
\item \emph{All permutations} (Section~\ref{sec:sym}). Cells recover all
Kolmogorov models, cells of cheap partitions exactly the strong models;
but the information of a complex model sits in the placement of its
certificate.
\item \emph{Linear symmetry} (Section~\ref{sec:GL}). The cells of
$\Par_{\GL}$ are the linearly homogeneous sets; affine subspaces are cells
of cheap symmetric partitions.
\item \emph{A band and its extremes}
(Sections~\ref{sec:band}--\ref{sec:extremes}).
$C(x)-\alpha\le\hGL_x(\alpha)\le n-\alpha$, both edges are attained, and
there are stochastic normal strings with $\soph(x)\approx0$ but
$\soph^{\GL}(x)\approx C(x)$.
\item \emph{Coordinates} (Section~\ref{sec:coords}). Every permutation
group is a \emph{type} (a Burnside ring element) with a \emph{placement}
(a permutation). Coordinates small enough to search are small enough to
miss simple structure.
\item \emph{Algorithms} (Section~\ref{sec:algorithms}). Orbit-model search
as a decision process on subgroup lattices: the uncosted problem
collapses onto the stabiliser of $x$, and with the cost of the hypothesis
included, model search is value estimation, with computable bounds,
certified training targets and a lattice version of Monte Carlo tree
search.
\end{enumerate}
Figures~\ref{fig:cells} and~\ref{fig:gl3types} illustrate cells, types
and placements on the $3$-cube.

\subsection{From theory to search}
\label{sec:programme}

Section~\ref{sec:algorithms} develops learned search, in particular
reinforcement learning, over orbit models. The theory of the earlier
sections is meant to make that search well posed. The coordinates of
Section~\ref{sec:coords}
give the state and action spaces: a hypothesis is a type together with a
placement, and the moves are relabelling, restriction with its Mackey
splitting of cells, and meets with other cheap hypotheses. The symmetric
structure function gives the objective, through the deficiency
$C(p,B)+\log|B|-C(x)$ of the model found, with $C(x)$ estimated by a
compressor. The maximal gap theorem says what failure means: a search
that finds no good orbit model has shown that $x$ has little symmetric
structure, not that it is random (Section~\ref{sec:search}).

\section{Preliminaries}

\subsection{Structure functions and sophistication}

$C$ is plain Kolmogorov complexity and $\CT(y\mid x)$ the \emph{total}
conditional complexity: the length of a shortest program that maps $x$
to $y$ and halts on every input. Strings have length $n$. Statements
said to hold ``up to $O(\log n)$'' hold up to additive $O(\log n)$ terms
in all coordinates; $c$ denotes a sufficiently large constant.

\begin{definition}
The \emph{structure function} of $x\in\{0,1\}^n$ is
$h_x(\alpha)=\min\{\log|S| : x\in S\subseteq\{0,1\}^n,\ C(S)\le\alpha\}$
for $\alpha\ge c\log n$. A model is \emph{sufficient} if
$C(S)+\log|S|\le C(x)+c\log n$, and the \emph{sophistication}
$\soph(x)=\min\{\alpha : h_x(\alpha)+\alpha\le C(x)+c\log n\}$ is the
complexity of a minimal sufficient statistic~\cite{Ko87,GTV,VV04}. The
string is \emph{stochastic} if $\soph(x)=O(\log n)$, and
\emph{antistochastic} if $h_x(\alpha)\approx n-\alpha$ for all
$\alpha<C(x)$~\cite{M15}.
\end{definition}

\begin{lemma}[\cite{VV04,SUV}]\label{lem:basic}
For $x\in\{0,1\}^n$ with $C(x)=\kappa$, up to $O(\log n)$:
(i) $h_x(\alpha)\ge\kappa-\alpha$; (ii) $h_x(c\log n)\le n$ and
$h_x(\kappa)=0$; (iii) (\emph{halving law}) $h_x(\alpha)+\alpha$ is
non-increasing. Conversely~\cite{VV04}, every simple function with
(i)--(iii) is the structure function of some $x$ up to $O(\log n)$; such
functions are called \emph{admissible}.
\end{lemma}

Figure~\ref{fig:sf} shows a typical structure function. The curve
starts near $n$, where the only affordable model is essentially
$\{0,1\}^n$, and reaches $0$ at $\alpha\approx C(x)$, where $\{x\}$ is
affordable. It never lies below the \emph{sufficiency line}
$C(x)-\alpha$ (Lemma~\ref{lem:basic}(i)), and by the halving law it falls
by at least one bit per additional bit of model complexity, possibly in
sudden drops. It first meets the sufficiency line at the minimal
sufficient statistic, of complexity $\soph(x)$, and follows the line
from there on. The region above the curve is the \emph{profile}
$P_x$ of all pairs $(\alpha,\log|S|)$ achieved by models $S\ni x$. For a
stochastic string the curve drops to the sufficiency line almost at once;
for an antistochastic string it follows $n-\alpha$ until $\alpha\approx
C(x)$.

\begin{figure}[t]
\centering
\begin{tikzpicture}[x=1.65cm,y=0.46cm]
\fill[blue!8] (0,10) -- (1.5,8.5) -- (1.5,6.5) -- (2.5,5.5) -- (2.5,3.5)
  -- (6,0) -- (7,0) -- (7,10) -- cycle;
\draw[->] (0,0) -- (7.3,0) node[below left]{$\alpha$};
\draw[->] (0,0) -- (0,10.8) node[above right]{$\log|S|$ of best model};
\draw[dashed,gray] (0,6) -- (6,0);
\draw[dotted,gray] (2.5,0) -- (2.5,3.5);
\draw[very thick,blue!70!black] (0,10) -- (1.5,8.5) -- (1.5,6.5)
  -- (2.5,5.5) -- (2.5,3.5) -- (6,0) -- (7,0);
\fill[blue!70!black] (2.5,3.5) circle (2.2pt);
\node[left] at (0,10) {$n$};
\node[left] at (0,6) {$C(x)$};
\node[below] at (0,0) {$0$};
\node[below] at (2.5,0) {$\soph(x)$};
\node[below] at (6,0) {$C(x)$};
\node[right] at (1.55,7.6) {$h_x(\alpha)$};
\node[right] at (2.6,3.7) {\small minimal sufficient statistic};
\node at (5,7.2) {\small profile $P_x$};
\node[gray!60!black,align=left,anchor=west] at (0.1,1.8) {\small sufficiency line\\ \small $C(x)-\alpha$};
\draw[gray!60!black,->,shorten >=2pt] (0.55,2.6) -- (0.8,5.1);
\end{tikzpicture}
\caption{A typical structure function (schematic, after~\cite{VV04}). It
lies above the sufficiency line, decreases with slope at most $-1$, and
touches the line first at the minimal sufficient statistic.}
\label{fig:sf}
\end{figure}

\FloatBarrier
\subsection{Strong models and partitions}

\begin{definition}[\cite{V15,M16}]
A model $A\ni x$ is \emph{$\varepsilon$-strong} if
$\CT(A\mid x)\le\varepsilon$, and
$\hstr_{x,\varepsilon}(\alpha)=\min\{\log|A| : x\in A,\
\CT(A\mid x)\le\varepsilon,\ C(A)\le\alpha\}$. A string is
\emph{$(\varepsilon,\delta)$-normal} if
$\hstr_{x,\varepsilon}(\alpha+\delta)\le h_x(\alpha)+\delta$ for all
$\alpha$, and \emph{strange} if it has no strong minimal sufficient
statistic.
\end{definition}

Strong models exclude the universal models of~\cite{GTV}, which carry
the information of the number of strings of bounded
complexity~\cite{V15}. Strange strings exist~\cite{V15}; normal strings
exist for every admissible shape~\cite{M16}; antistochastic strings are
normal~\cite{M16}. The following lemma is the bridge to partitions.

\begin{lemma}[Partition lemma~\cite{M16}]\label{lem:partition}
If $A\ni x$ is $\varepsilon$-strong, there is a partition $p$ of
$\{0,1\}^n$ with $C(p)\le\varepsilon+O(\log n)$ and a cell
$A_1\ni x$ of $p$ with $A_1\subseteq A$ and
$C(A_1)\le C(A)+\varepsilon+O(\log n)$.
\end{lemma}

\begin{proof}
Let $r$ be a total program with $r(x)=A$ and $|r|\le\varepsilon$. Group
the strings $x'$ with $x'\in r(x')$ according to the value $r(x')$, and
put every other string in a singleton cell. This partition is
computable from $r$ and $n$, and the cell of $x$ is
$A_1=\{x''\in A : r(x'')=A\}$, computable from $A$, $r$ and $n$.
\end{proof}

\subsection{Restricted model classes}
\label{sec:prelim-restricted}

For a family $\A$ of finite sets, the \emph{restricted structure
function} $h^{\A}_x(\alpha)=\min\{\log|A| : x\in A\in\A,\ C(A)\le\alpha\}$
uses only models from $\A$, and $\soph^{\A}(x)$ is defined from it as
$\soph(x)$ is from $h_x$. The family is \emph{acceptable}~\cite{VV10}
(see also~\cite{M16,SUV}) if it is enumerable, contains every
$\{0,1\}^n$, and satisfies the \emph{covering condition}: for a fixed
polynomial $\rho$, every $A\in\A$ and every $c<|A|$, the $n$-bit strings
of $A$ are covered by at most $\rho(n)|A|/c$ members of $\A$ of size at
most $c$.

\begin{theorem}[\cite{VV10}]\label{thm:VV}
If $\A$ is acceptable, $h^{\A}_x$ obeys the halving law, and for every
admissible shape $P$ there is $x$ with $h_x$ and $h^{\A}_x$ both
$C(P)+O(\sqrt{n\log n})$-close to $P$.
\end{theorem}

\section{The Galois connection}
\label{sec:galois}\label{sec:sympart}

\subsection{Symmetric partitions}

Fix an \emph{ambient} finite group $G$ acting on a finite set $X$, with
$G=G_n$ computable from $n$ when $X=\{0,1\}^n$. Let $L(G)$ be the
subgroup lattice and $\Par(X)$ the partition lattice, ordered by
refinement ($p\le q$ if every cell of $p$ lies in a cell of $q$). Put
\[
  \orb_G(U)=\{\text{orbits of }U\},\qquad
  \stab_G(p)=\{g\in G : g(B)=B\text{ for every cell }B\in p\}.
\]
\begin{proposition}[\cite{Ke99}]\label{prop:galois}
$\orb_G(U)\le p\iff U\le\stab_G(p)$. Hence both maps are monotone,
$U\le\stab_G\orb_G(U)$, $\orb_G\stab_G(p)\le p$, and the \emph{closed
subgroups} ($U=\stab_G\orb_G U$) correspond bijectively and
order-preservingly to the \emph{symmetric partitions}
($p=\orb_G\stab_G p$), which are exactly the orbit partitions of
subgroups of $G$. Denote their lattice by $\Par_G$.
\end{proposition}

\begin{proof}
If every $U$-orbit lies in a cell of $p$, each cell is a union of
$U$-orbits and is $U$-invariant; conversely, if $U$ preserves every
cell its orbits lie in cells. The rest is the general theory of
Galois connections.
\end{proof}

The closed subgroup $\stab_G(p)$ is the \emph{largest} subgroup with
orbit partition $p$: the \emph{canonical certificate} of the hypothesis
$p$.

\begin{example}\label{ex:closed}
\leavevmode
\begin{enumerate}[label=(\alph*),nosep]
\item $G=\Sym(X)$: every partition is symmetric, and the canonical
certificates are the \emph{Young subgroups} $\prod_{B\in p}\Sym(B)$.
\item $G=S_n$ permuting the coordinates of $\{0,1\}^n$: $\orb(S_n)$ is
the partition into weight classes; subgroups give finer
``combinatorial'' partitions.
\item $G=\GL(n,2)$: symmetric partitions are rigid. If
$W\ne\F_2^n$ is a subspace with $\dim W\ge2$, the partition
$\{W\setminus\{0\}\}\cup\{\text{singletons}\}$ is not symmetric: any
$g$ preserving it fixes every vector outside $W\setminus\{0\}$, these
span $\F_2^n$, so $g=1$. With a canonical complement $D$, the orbit
partition of $\GL(W)\times\{\mathrm{id}_D\}$ is symmetric; its cells
are $\{d\}$ and $d+(W\setminus\{0\})$, $d\in D$. A cell can appear
only together with the other orbits of its certificate.
\end{enumerate}
\end{example}

\subsection{Canonical cost}

$C(p)$ is the complexity of a partition $p$ of $\{0,1\}^n$ as a finite
object, $C(p,B)$ that of $p$ with one of its cells, and $C(U)$ the
complexity of a generator list of $U$.

\begin{proposition}[Canonical cost]\label{prop:canon}
For $U\le G$ with $p=\orb_G(U)$ and a cell $B$ of $p$:
\[
  C(\stab_G p)\le C(p)+O(1)\le C(U)+O(1),\qquad
  C(\stab_G p,B)=C(p,B)\pm O(1).
\]
The canonical certificate is, up to $O(1)$, the cheapest group with the
given orbits, and it costs exactly as much as the hypothesis.
\end{proposition}

\begin{proof}
Orbits are computable from generators; $\stab_G(p)$ is computable from
$p$ by search in $G_n$, and $n$ is recoverable from $p$.
\end{proof}

\subsection{Cells: homogeneous sets}
\label{sec:cells}

A structure function sees a hypothesis only through the cell containing
$x$ and the cost of naming it. So the basic question is which sets occur
as cells. Write $G_{\{S\}}=\{g\in G: g(S)=S\}$ for the setwise
stabiliser of $S\subseteq X$.

\begin{definition}
A set $S\subseteq X$ is \emph{$G$-homogeneous} if $G_{\{S\}}$ acts
transitively on $S$. Let $\A_G$ denote the family of non-empty
$G$-homogeneous subsets of $X$.
\end{definition}

\begin{proposition}[Cells]\label{prop:cells}
Let $S\subseteq X$ be non-empty.
\begin{enumerate}[label=(\roman*),nosep]
\item $S$ is a cell of some partition in $\Par_G$ if and only if
$S\in\A_G$.
\item If $S\in\A_G$, then $p_S=\orb_G(G_{\{S\}})$ is the coarsest
partition in $\Par_G$ having $S$ as a cell: every $p\in\Par_G$ with
$S\in p$ refines $p_S$.
\item If $S\in\A_G$, then $C(p_S,S)=C(S)\pm O(1)$. For every partition
$p$ having $S$ as a cell, $C(S)\le C(p,S)+O(1)$.
\end{enumerate}
\end{proposition}

\begin{proof}
(i) If $S=U\cdot x$ for some $U\le G$, then $U\le G_{\{S\}}$, so
$G_{\{S\}}$ is transitive on $S$. Conversely, if $G_{\{S\}}$ is
transitive on $S$, then $S$ is one of its orbits, hence a cell of
$p_S$, which is symmetric by Proposition~\ref{prop:galois}.
(ii) If $p\in\Par_G$ and $S\in p$, then $\stab_G(p)\le G_{\{S\}}$, so
$p=\orb_G\stab_G(p)\le\orb_G(G_{\{S\}})=p_S$ because $\orb_G$ is
monotone.
(iii) $G_{\{S\}}$, and hence $p_S$, is computable from $S$ by search in
$G_n$, and $n$ is recoverable from $S$. Conversely $S$ is computable
from $(p,S)$.
\end{proof}

So a cell carries its own canonical hypothesis, and naming the cell
already names that hypothesis. Whole symmetric partitions are rigid
(Example~\ref{ex:closed}(c)), but single cells are constrained only by
homogeneity.

\begin{example}\label{ex:cells}
\leavevmode
\begin{enumerate}[label=(\alph*),nosep]
\item $G=\Sym(X)$: every non-empty set is homogeneous, and $p_S$ is
$\{S,X\setminus S\}$ (or $\{X\}$ if $S=X$).
\item $G=\GL(n,2)$: a homogeneous set containing $0$ is $\{0\}$, since
$0$ is fixed. Theorem~\ref{thm:homog} identifies the others.
\item $G=\GL(3,2)$: $79$ of the $255$ non-empty subsets of $\F_2^3$ are
homogeneous, distributed by size $1,\dots,8$ as
$8,21,35,7,0,7,1,0$. Every pair and every triple of nonzero vectors is
homogeneous, and no $5$-set is. Only $100$ of the $4140$ partitions of
$\F_2^3$ are symmetric, but $695$ have every cell homogeneous. See Figure~\ref{fig:cells}.
\end{enumerate}
\end{example}

\begin{figure}[t]
\centering
\setlength{\tabcolsep}{6pt}
\begin{tabular}{cccccc}
\begin{tikzpicture}[x=1.08cm,y=1.08cm]
\draw[cubeedge] (0.00,0.00)--(1.00,0.00);
\draw[cubeedge] (0.00,0.00)--(0.00,1.00);
\draw[cubeedge] (0.00,0.00)--(0.50,0.40);
\draw[cubeedge] (1.00,0.00)--(1.00,1.00);
\draw[cubeedge] (1.00,0.00)--(1.50,0.40);
\draw[cubeedge] (0.00,1.00)--(1.00,1.00);
\draw[cubeedge] (0.00,1.00)--(0.50,1.40);
\draw[cubeedge] (1.00,1.00)--(1.50,1.40);
\draw[cubeedge] (0.50,0.40)--(1.50,0.40);
\draw[cubeedge] (0.50,0.40)--(0.50,1.40);
\draw[cubeedge] (1.50,0.40)--(1.50,1.40);
\draw[cubeedge] (0.50,1.40)--(1.50,1.40);
\filldraw[fixpt] (0.00,0.00) circle (2.6pt);
\draw[ring] (1.00,0.00) circle (5.2pt);
\fill[wA] (1.00,0.00) circle (3pt);
\draw[ring] (0.00,1.00) circle (5.2pt);
\fill[wA] (0.00,1.00) circle (3pt);
\filldraw[fixpt] (1.00,1.00) circle (2.6pt);
\fill[wB] (0.50,0.40) circle (3pt);
\fill[wB] (1.50,0.40) circle (3pt);
\fill[wB] (0.50,1.40) circle (3pt);
\fill[wB] (1.50,1.40) circle (3pt);
\end{tikzpicture}
&
\begin{tikzpicture}[x=1.08cm,y=1.08cm]
\draw[cubeedge] (0.00,0.00)--(1.00,0.00);
\draw[cubeedge] (0.00,0.00)--(0.00,1.00);
\draw[cubeedge] (0.00,0.00)--(0.50,0.40);
\draw[cubeedge] (1.00,0.00)--(1.00,1.00);
\draw[cubeedge] (1.00,0.00)--(1.50,0.40);
\draw[cubeedge] (0.00,1.00)--(1.00,1.00);
\draw[cubeedge] (0.00,1.00)--(0.50,1.40);
\draw[cubeedge] (1.00,1.00)--(1.50,1.40);
\draw[cubeedge] (0.50,0.40)--(1.50,0.40);
\draw[cubeedge] (0.50,0.40)--(0.50,1.40);
\draw[cubeedge] (1.50,0.40)--(1.50,1.40);
\draw[cubeedge] (0.50,1.40)--(1.50,1.40);
\filldraw[fixpt] (0.00,0.00) circle (2.6pt);
\draw[ring] (1.00,0.00) circle (5.2pt);
\fill[wA] (1.00,0.00) circle (3pt);
\draw[ring] (0.00,1.00) circle (5.2pt);
\fill[wA] (0.00,1.00) circle (3pt);
\draw[ring] (1.00,1.00) circle (5.2pt);
\fill[wA] (1.00,1.00) circle (3pt);
\fill[wB] (0.50,0.40) circle (3pt);
\fill[wB] (1.50,0.40) circle (3pt);
\fill[wB] (0.50,1.40) circle (3pt);
\fill[wB] (1.50,1.40) circle (3pt);
\end{tikzpicture}
&
\begin{tikzpicture}[x=1.08cm,y=1.08cm]
\draw[cubeedge] (0.00,0.00)--(1.00,0.00);
\draw[cubeedge] (0.00,0.00)--(0.00,1.00);
\draw[cubeedge] (0.00,0.00)--(0.50,0.40);
\draw[cubeedge] (1.00,0.00)--(1.00,1.00);
\draw[cubeedge] (1.00,0.00)--(1.50,0.40);
\draw[cubeedge] (0.00,1.00)--(1.00,1.00);
\draw[cubeedge] (0.00,1.00)--(0.50,1.40);
\draw[cubeedge] (1.00,1.00)--(1.50,1.40);
\draw[cubeedge] (0.50,0.40)--(1.50,0.40);
\draw[cubeedge] (0.50,0.40)--(0.50,1.40);
\draw[cubeedge] (1.50,0.40)--(1.50,1.40);
\draw[cubeedge] (0.50,1.40)--(1.50,1.40);
\filldraw[fixpt] (0.00,0.00) circle (2.6pt);
\fill[wA] (1.00,0.00) circle (3pt);
\fill[wA] (0.00,1.00) circle (3pt);
\fill[wA] (1.00,1.00) circle (3pt);
\draw[ring] (0.50,0.40) circle (5.2pt);
\fill[wB] (0.50,0.40) circle (3pt);
\draw[ring] (1.50,0.40) circle (5.2pt);
\fill[wB] (1.50,0.40) circle (3pt);
\draw[ring] (0.50,1.40) circle (5.2pt);
\fill[wB] (0.50,1.40) circle (3pt);
\draw[ring] (1.50,1.40) circle (5.2pt);
\fill[wB] (1.50,1.40) circle (3pt);
\end{tikzpicture}
&
\begin{tikzpicture}[x=1.08cm,y=1.08cm]
\draw[cubeedge] (0.00,0.00)--(1.00,0.00);
\draw[cubeedge] (0.00,0.00)--(0.00,1.00);
\draw[cubeedge] (0.00,0.00)--(0.50,0.40);
\draw[cubeedge] (1.00,0.00)--(1.00,1.00);
\draw[cubeedge] (1.00,0.00)--(1.50,0.40);
\draw[cubeedge] (0.00,1.00)--(1.00,1.00);
\draw[cubeedge] (0.00,1.00)--(0.50,1.40);
\draw[cubeedge] (1.00,1.00)--(1.50,1.40);
\draw[cubeedge] (0.50,0.40)--(1.50,0.40);
\draw[cubeedge] (0.50,0.40)--(0.50,1.40);
\draw[cubeedge] (1.50,0.40)--(1.50,1.40);
\draw[cubeedge] (0.50,1.40)--(1.50,1.40);
\draw[ring] (0.00,0.00) circle (5.2pt);
\filldraw[fixpt] (0.00,0.00) circle (2.6pt);
\draw[ring] (1.00,0.00) circle (5.2pt);
\fill[wA] (1.00,0.00) circle (3pt);
\draw[ring] (0.00,1.00) circle (5.2pt);
\fill[wA] (0.00,1.00) circle (3pt);
\filldraw[fixpt] (1.00,1.00) circle (2.6pt);
\fill[wB] (0.50,0.40) circle (3pt);
\fill[wB] (1.50,0.40) circle (3pt);
\fill[wB] (0.50,1.40) circle (3pt);
\fill[wB] (1.50,1.40) circle (3pt);
\end{tikzpicture}
&
\begin{tikzpicture}[x=1.08cm,y=1.08cm]
\draw[cubeedge] (0.00,0.00)--(1.00,0.00);
\draw[cubeedge] (0.00,0.00)--(0.00,1.00);
\draw[cubeedge] (0.00,0.00)--(0.50,0.40);
\draw[cubeedge] (1.00,0.00)--(1.00,1.00);
\draw[cubeedge] (1.00,0.00)--(1.50,0.40);
\draw[cubeedge] (0.00,1.00)--(1.00,1.00);
\draw[cubeedge] (0.00,1.00)--(0.50,1.40);
\draw[cubeedge] (1.00,1.00)--(1.50,1.40);
\draw[cubeedge] (0.50,0.40)--(1.50,0.40);
\draw[cubeedge] (0.50,0.40)--(0.50,1.40);
\draw[cubeedge] (1.50,0.40)--(1.50,1.40);
\draw[cubeedge] (0.50,1.40)--(1.50,1.40);
\filldraw[fixpt] (0.00,0.00) circle (2.6pt);
\draw[ring] (1.00,0.00) circle (5.2pt);
\fill[wA] (1.00,0.00) circle (3pt);
\draw[ring] (0.00,1.00) circle (5.2pt);
\fill[wA] (0.00,1.00) circle (3pt);
\draw[ring] (1.00,1.00) circle (5.2pt);
\fill[wA] (1.00,1.00) circle (3pt);
\draw[ring] (0.50,0.40) circle (5.2pt);
\filldraw[fixpt] (0.50,0.40) circle (2.6pt);
\fill[wB] (1.50,0.40) circle (3pt);
\fill[wB] (0.50,1.40) circle (3pt);
\fill[wB] (1.50,1.40) circle (3pt);
\end{tikzpicture}
&
\begin{tikzpicture}[x=1.08cm,y=1.08cm]
\draw[cubeedge] (0.00,0.00)--(1.00,0.00);
\draw[cubeedge] (0.00,0.00)--(0.00,1.00);
\draw[cubeedge] (0.00,0.00)--(0.50,0.40);
\draw[cubeedge] (1.00,0.00)--(1.00,1.00);
\draw[cubeedge] (1.00,0.00)--(1.50,0.40);
\draw[cubeedge] (0.00,1.00)--(1.00,1.00);
\draw[cubeedge] (0.00,1.00)--(0.50,1.40);
\draw[cubeedge] (1.00,1.00)--(1.50,1.40);
\draw[cubeedge] (0.50,0.40)--(1.50,0.40);
\draw[cubeedge] (0.50,0.40)--(0.50,1.40);
\draw[cubeedge] (1.50,0.40)--(1.50,1.40);
\draw[cubeedge] (0.50,1.40)--(1.50,1.40);
\filldraw[fixpt] (0.00,0.00) circle (2.6pt);
\draw[ring] (1.00,0.00) circle (5.2pt);
\filldraw[fixpt] (1.00,0.00) circle (2.6pt);
\draw[ring] (0.00,1.00) circle (5.2pt);
\fill[wA] (0.00,1.00) circle (3pt);
\draw[ring] (1.00,1.00) circle (5.2pt);
\fill[wA] (1.00,1.00) circle (3pt);
\draw[ring] (0.50,0.40) circle (5.2pt);
\fill[wA] (0.50,0.40) circle (3pt);
\draw[ring] (1.50,0.40) circle (5.2pt);
\fill[wA] (1.50,0.40) circle (3pt);
\fill[wB] (0.50,1.40) circle (3pt);
\fill[wB] (1.50,1.40) circle (3pt);
\end{tikzpicture}
\\
\small \{100,010\} & \small a line & \small affine plane & \small contains $000$ & \small $e_1,e_2,e_1{+}e_2,e_3$ & \small a $5$-set\\
\small\textsf{homogeneous} & \small\textsf{homogeneous} & \small\textsf{homogeneous} & \small\textsf{not homogeneous} & \small\textsf{not homogeneous} & \small\textsf{not homogeneous}\\
\end{tabular}
\caption{Cells of $\GL(3,2)$ on the $3$-cube. The ringed corners form a set $S$; colours are the orbits of its setwise stabiliser $\GL(3,2)_{\{S\}}$, and hollow corners are fixed points. $S$ is a cell of a symmetric partition exactly when it is one colour (Proposition~\ref{prop:cells}). The fifth set is the non-homogeneous example of Example~\ref{ex:homog}. An interactive version is in the HTML edition.}
\label{fig:cells}
\end{figure}

\subsection{The algebra of hypotheses}
\label{sec:algebra}

Since $\orb_G$ is a lower adjoint it preserves joins, and $\stab_G$ as an
upper adjoint preserves meets:
\[
  \orb_G(\langle U,V\rangle)=\orb_G(U)\vee\orb_G(V),\qquad
  \stab_G(p\wedge q)=\stab_G(p)\cap\stab_G(q).
\]
So $\Par_G$ is closed under joins in $\Par(X)$: the \emph{join} of two
hypotheses is the hypothesis of all their symmetries together, and is
coarser. The meet in $\Par_G$ is
$p\wedge_G q=\orb_G(\stab_G p\cap\stab_G q)\le p\wedge q$, the
hypothesis of their \emph{common} symmetries, which is finer. Both are
computable, so $C(p\vee q),C(p\wedge_G q)\le C(p)+C(q)+O(\log n)$. The
cell of $x$ in $p\wedge_G q$ lies inside the intersection of its
cells in $p$ and $q$. Meeting with a second cheap hypothesis is thus
the natural way to sharpen an explanation (compare the paragraph on
halving at the end of Section~\ref{sec:band} and
Question~\ref{q:accept}).

\subsection{Relabelling}

$G$ acts by conjugation on $L(G)$ and by translation on $\Par_G$, and
the two maps are equivariant:
\[
  \orb_G(gUg^{-1})=g\cdot\orb_G(U),\qquad
  \stab_G(g\cdot p)=g\,\stab_G(p)\,g^{-1}.
\]
The model $(g\cdot p,gB)$
for $gx$ is the model $(p,B)$ for $x$ in new coordinates, and costs at
most $C(g)+O(1)$ more.

\section{Partition structure functions}
\label{sec:partsf}

\begin{definition}
Let $\cP$ assign to each $n$ a family of partitions of $\{0,1\}^n$. A
\emph{$\cP$-model} for $x$ is a pair $(p,B)$ with $p\in\cP$ and
$x\in B\in p$. Define
\begin{align*}
h^{\cP}_x(\alpha)&=\min\{\log|B| : (p,B)\ \cP\text{-model},\ C(p,B)\le\alpha\},\\
h^{\cP}_{x,\varepsilon}(\alpha)&=\min\{\log|B| : (p,B)\ \cP\text{-model},\
C(p)\le\varepsilon,\ C(p,B)\le\alpha\},\\
\soph^{\cP}(x)&=\min\{\alpha : h^{\cP}_x(\alpha)+\alpha\le C(x)+c\log n\}.
\end{align*}
$C(p)$ is the cost of the \emph{hypothesis} and $C(B\mid p)$ that of
the \emph{cell}; $C(p,B)=C(p)+C(B\mid p)$ up to $O(\log n)$. For an
ambient group $G$ we write $h^G_x$, $h^G_{x,\varepsilon}$ and
$\soph^G$ for $\cP=\Par_G$.
\end{definition}

\begin{proposition}\label{prop:basicP}
Up to $O(1)$ shifts:
(i) $h_x\le h^{\cP}_x\le h^{\cP}_{x,\varepsilon}$, and $\cP\subseteq\cP'$
implies $h^{\cP'}_x\le h^{\cP}_x$;
(ii) every cell of a partition of complexity at most $\varepsilon$ is
$(\varepsilon+O(1))$-strong, so
$\hstr_{x,\varepsilon+c}\le h^{\cP}_{x,\varepsilon}$;
(iii) $\soph(x)\le\soph^{\cP}(x)$ up to $O(\log n)$.
\end{proposition}

\begin{proof}
(i) A cell is a finite set containing $x$ and $C(B)\le C(p,B)+O(1)$.
(ii) A total program knowing $p$ maps $x$ to its cell.
(iii) Follows from (i).
\end{proof}

\begin{corollary}[Cells suffice]\label{cor:cells}
For every ambient $G$, up to $O(1)$ shifts in $\alpha$,
\[
  h^G_x=h^{\A_G}_x\qquad\text{and}\qquad \soph^G(x)=\soph^{\A_G}(x).
\]
\end{corollary}

\begin{proof}
A $\Par_G$-model $(p,B)$ has $B\in\A_G$ and $C(B)\le C(p,B)+O(1)$. A set
$B\in\A_G$ containing $x$ gives the $\Par_G$-model $(p_B,B)$ with
$C(p_B,B)\le C(B)+O(1)$. Both by Proposition~\ref{prop:cells}.
\end{proof}

The strong version does not reduce to cells in the same way. The
canonical hypothesis satisfies $C(p_B)\le C(B)+O(1)$, but this can exceed
the cost of the cheapest symmetric partition containing $B$ by up to
$C(B\mid p)\le\log|p|\le n$ bits. So $h^G_{x,\varepsilon}$ depends on
which cheap partitions carry $B$, not only on $B$
(compare Question~\ref{q:rigid}).

\paragraph{Groups and partitions.} The group formulation defines
$h^{L(G)}_x(\alpha)$ as the least $\log|U\cdot x|$ over subgroups
$U\le G$ with $C(U,U\cdot x)\le\alpha$, and $h^{L(G)}_{x,\varepsilon}$
by adding $C(U)\le\varepsilon$.

\begin{theorem}[Dictionary]\label{thm:dict}
For every ambient $G$, up to $O(1)$ shifts in $\alpha$ and $\varepsilon$,
$h^{L(G)}_x=h^G_x$ and $h^{L(G)}_{x,\varepsilon}=h^G_{x,\varepsilon}$.
\end{theorem}

\begin{proof}
A pair $(U,U\cdot x)$ gives the $\Par_G$-model $(\orb_G U,U\cdot x)$ at no
greater cost. A $\Par_G$-model $(p,B)$ gives $(\stab_G p,B)$, and
$\stab_G(p)\cdot x=B$ because $p$ is symmetric. Costs are compared by
Proposition~\ref{prop:canon}.
\end{proof}

\begin{proposition}[Burnside form]\label{prop:burnside}
If $(p,B)$ is a $\Par_G$-model for $x$ with certificate $U=\stab_G(p)$,
then $|B|=|U|/|\stab_U(x)|=1/\Pr_{u\in U}[ux=x]$, so
$C(x\mid p,B)\le\log|B|+O(1)$. All but a $2^{-t}$ fraction of
$x'\in B$ have $C(x'\mid p,B)\ge\log|B|-t$.
\end{proposition}

\begin{center}
\begin{tabular}{ll}
\toprule
symmetry & algorithmic statistics\\
\midrule
symmetric partition $p\in\Par_G$ & hypothesis (model class)\\
cell $B\ni x$ & model\\
$\log|U|-\log|\stab_U(x)|$ & $\log|B|$\\
$C(p)$ small & strong model\\
$x$ typical in its cell & small randomness deficiency\\
$\soph^G(x)$ & $\soph(x)$\\
\bottomrule
\end{tabular}
\end{center}

\section{All permutations}
\label{sec:sym}

Let $\Sym$ denote the ambient group $\Sym(\{0,1\}^n)$, written $S_X$
with $X=\{0,1\}^n$ in Section~\ref{sec:coords}. By
Example~\ref{ex:closed}(a), $\Par_{\Sym}$ is the lattice of \emph{all}
partitions, and by Example~\ref{ex:cells}(a) every non-empty set is a
cell.

\begin{theorem}\label{thm:sym}
Up to $O(1)$ shifts, $h^{\Sym}_x=h_x$, so $\soph^{\Sym}=\soph$. Up to
$O(\varepsilon+\log n)$ shifts, $h^{\Sym}_{x,\varepsilon}=\hstr_{x,\varepsilon}$.
\end{theorem}

\begin{proof}
Every set is $\Sym$-homogeneous, so the first claim is
Corollary~\ref{cor:cells}; explicitly, for $S\ni x$ the partition
$\{S,X\setminus S\}$ costs $C(S)+O(1)$. For the second claim combine
Proposition~\ref{prop:basicP}(ii) with the partition lemma
(Lemma~\ref{lem:partition}).
\end{proof}

\begin{corollary}
Up to $O(\varepsilon+\delta+\log n)$, a string is
$(\varepsilon,\delta)$-normal if and only if cells of partitions of
complexity at most $\varepsilon$ (equivalently, orbits of groups of
complexity at most $\varepsilon$) are as good as arbitrary models at
every level. It is strange if and only if no such cell is a minimal
sufficient statistic. So there are strings with no sufficient cheap
symmetry explanation~\cite{V15}, and every admissible shape is realised
by a string whose cheap-symmetry structure function equals its
Kolmogorov structure function~\cite{M16}.
\end{corollary}

Theorem~\ref{thm:sym} says that unrestricted symmetry explains
everything. As a statement about symmetry this is nearly vacuous. The
canonical certificates in $S_X$ are Young subgroups, whose type records
only the sizes of the cells. In the coordinates of
Section~\ref{sec:coords}, a model of complexity much larger than
$2^{n/2}$ has all but $O(2^{n/2})$ bits of its information in the
placement of its certificate, under every realisation
(Proposition~\ref{prop:placement}). Moreover, the canonical certificate
in $S_X$ can hide a genuinely structured symmetry
(Example~\ref{ex:weight}).

\section{Linear symmetry}
\label{sec:GL}

Identify $\{0,1\}^n$ with $\F_2^n$ and take $G=\GL(n,2)$. Linear maps
fix $0$, so throughout $x\ne0$.

\subsection{Cells: linearly homogeneous sets}

Which sets can be cells of linear-symmetric partitions? A linear map
preserves all $\F_2$-linear dependencies, so a set can be an orbit of a
linear group only if its internal dependency structure looks the same
from each of its points. We make this precise.

A \emph{relation} in $S\subseteq\F_2^n\setminus\{0\}$ is a subset
$T\subseteq S$ with $\sum T=0$, an XOR dependency among elements of $S$.
For example, in $S=\{e_1,e_2,e_1+e_2\}$ the whole set is a relation,
since $e_1+e_2+(e_1+e_2)=0$. The family of all relations is the
\emph{dependency structure} of $S$ (the binary matroid it represents). A
permutation $\pi$ of $S$ \emph{preserves relations} if
$\sum T=0\iff\sum\pi(T)=0$ for all $T\subseteq S$, i.e.\ if it is an
automorphism of the dependency structure.

\begin{definition}
A set $S\subseteq\F_2^n\setminus\{0\}$ is \emph{linearly homogeneous} if
the relation-preserving permutations of $S$ act transitively on $S$:
for any $s,s'\in S$ some automorphism of the dependency structure maps
$s$ to $s'$.
\end{definition}

The analogy is with vertex-transitive graphs, where every vertex looks
alike; here every point looks alike with respect to XOR dependencies
instead of edges. Two extreme cases show the range of the notion. A
linearly independent set has no relations, so every permutation
preserves them and the set is homogeneous. A punctured subspace
$W\setminus\{0\}$ has as many relations as possible, but $\GL(W)$ moves
any nonzero vector to any other while preserving them, so it too is
homogeneous. Homogeneity fails when some points play special roles.

\begin{example}\label{ex:homog}
Linearly homogeneous sets include: $W\setminus\{0\}$ for a subspace
$W$; affine subspaces $a+W$ with $a\notin W$; linearly independent
sets; weight classes and other orbits of coordinate permutation groups;
and cosets of multiplicative subgroups of $\F_{2^n}^\times$, identifying
$\F_2^n$ with $\F_{2^n}$. The set $\{e_1,e_2,e_1+e_2,e_3\}$ is not
homogeneous: $e_1$ lies in the three-term relation
$\{e_1,e_2,e_1+e_2\}$ while $e_3$ lies in no relation, so no
relation-preserving permutation maps $e_1$ to $e_3$.
\end{example}

The key fact is a small piece of linear algebra: \emph{a permutation of
$S$ extends to an invertible linear map if and only if it preserves
relations}. This gives the characterisation of cells.

\begin{theorem}\label{thm:homog}
Let $S\subseteq\F_2^n\setminus\{0\}$. Every relation-preserving
permutation of $S$ is the restriction of an element of
$\GL(n,2)_{\{S\}}$. Hence $S$ is $\GL$-homogeneous if and only if it is
linearly homogeneous, and $\A_{\GL}$ consists of $\{0\}$ and the
linearly homogeneous sets. By Corollary~\ref{cor:cells}, $\hGL_x$ is
Kolmogorov's structure function restricted to linearly homogeneous sets.
\end{theorem}

\begin{proof}
The restriction of a linear map preserves relations. Conversely, a
relation-preserving $\pi$ defines a linear bijection
$L_\pi(\sum_{s\in T}s)=\sum_{s\in T}\pi(s)$ of $\mathrm{span}(S)$. It is
well defined, because $\sum T=\sum T'$ implies $\sum(T\triangle T')=0$
and hence $\sum\pi(T)=\sum\pi(T')$, and it is injective, because
$\pi^{-1}$ also preserves relations. Extending $L_\pi$ by the identity on
a complement of $\mathrm{span}(S)$ gives $g\in\GL(n,2)$ with $g(S)=S$
and $g|_S=\pi$. So the restrictions of $\GL(n,2)_{\{S\}}$ to $S$ are
exactly the relation-preserving permutations, and one group is
transitive on $S$ if and only if the other is. Since linear maps fix $0$,
a homogeneous set containing $0$ is $\{0\}$.
\end{proof}

\paragraph{Why this matters.} By Corollary~\ref{cor:cells} and
Theorem~\ref{thm:homog}, $\hGL_x$ is Vereshchagin and Vit\'anyi's
structure function for the concrete model class $\A_{\GL}$ of linearly
homogeneous sets~\cite{VV10}, so the groups drop out of the definition.
Since a homogeneous set determines its own certificate
(Proposition~\ref{prop:cells}), describing the group costs nothing extra,
and every difference between $\hGL_x$ and $h_x$ comes from \emph{which}
sets are admissible. Under $\Sym(\{0,1\}^n)$ every set qualifies
(Theorem~\ref{thm:sym}). Under $\GL(n,2)$ only homogeneous sets do, and
there are at most $2^{n^4+n}$ of them against $2^{2^n}$ arbitrary sets.
This scarcity is what the band (Section~\ref{sec:band}) measures and
what the counting proof of Theorem~\ref{thm:maxgap} exploits. Membership
in $\A_{\GL}$ is a finite test, namely transitivity of the setwise
stabiliser. However, half of a homogeneous set is usually not
homogeneous, so the halving law for $\hGL_x$ is a genuine open question
(Question~\ref{q:accept}).

Single cells are constrained only by homogeneity, but whole symmetric
partitions are rigid (Example~\ref{ex:closed}(c)). This matters for the
strong version $h^{\GL}_{x,\varepsilon}$, which requires a \emph{cheap}
symmetric partition containing the cell (Question~\ref{q:rigid}).
Linear hypotheses also have little room to hide information in
placement: a subgroup of $\GL(n,2)$ is determined by its type, an
abstract group with a faithful module, up to $n^2$ bits
(Proposition~\ref{prop:GLtype}).

\subsection{Affine hypotheses}

Every affine subspace is a cell of a cheap symmetric partition. For
$W\setminus\{0\}$ this is Example~\ref{ex:closed}(c). For $a+W$ with
$a\notin W$, fix a complement $D$ of $W\oplus\langle a\rangle$, and for
$w\in W$ let $\tau_w$ fix $W\oplus D$ pointwise and send $a\mapsto a+w$.
The group $\{\tau_w\}\cong W$ has cells $\{v\}$ for $v\in W\oplus D$ and
$a+d+W$ for $d\in D$, at cost $C(W,a)+O(1)$. Writing $\haff_x$ for the
structure function over affine subspaces, we get, up to $O(\log n)$,
\begin{equation}\label{eq:chains}
h_x\ \le\ \hGL_x\ \le\ \haff_x,\qquad
h_x\ \le\ \hstr_{x,\varepsilon}\ \le\ h^{\GL}_{x,\varepsilon}.
\end{equation}

\section{The band}
\label{sec:band}

\begin{theorem}[Band]\label{thm:band}
Let $x\in\F_2^n\setminus\{0\}$ with $C(x)=\kappa$. Up to $O(\log n)$:
\begin{enumerate}[label=(\roman*),nosep]
\item $\kappa-\alpha\le h_x(\alpha)\le\hGL_x(\alpha)\le\haff_x(\alpha)\le
n-\alpha$ for $c\log n\le\alpha\le\kappa$;
\item $\hGL_x(c\log n)\le\log(2^n-1)$ and $\hGL_x(\kappa)=0$;
\item $0\le\hGL_x(\alpha)-h_x(\alpha)\le n-\kappa$, and
$\soph(x)\le\soph^{\GL}(x)\le\kappa$;
\item (relabelling) for $g\in\GL(n,2)$,
$\hGL_{gx}(\alpha+C(g)+c)\le\hGL_x(\alpha)$.
\end{enumerate}
\end{theorem}

\begin{proof}
(i) By Lemma~\ref{lem:basic}(i) and~\eqref{eq:chains}; for the last
inequality use the affine subspace of strings agreeing with $x$ in the
first $j$ coordinates. (ii) Use the partitions $\{\{0\},\F_2^n\setminus\{0\}\}$
and the discrete partition. (iii) Follows from (i) and (ii). (iv) By
equivariance of the Galois connection.
\end{proof}

\paragraph{All the structure functions together.}
Figure~\ref{fig:all} places the structure functions of this paper in one
picture, for a typical \emph{normal} string $x$ with $\kappa=C(x)$; normal
strings are the common case~\cite{VS17}. Every curve lies between the
sufficiency line $\kappa-\alpha$ and the ceiling $n-\alpha$, starts at
most at $n$ and vanishes at $\alpha\approx\kappa$. Up to $O(1)$ shifts,
$h_x\le\hstr_{x,\varepsilon}\le h^{\GL}_{x,\varepsilon}$ and
$h_x\le\hGL_x\le h^{\GL}_{x,\varepsilon}$
(Proposition~\ref{prop:basicP}). The curve $h^{\Sym}_x$ lies on $h_x$ for
every string (Theorem~\ref{thm:sym}), and for normal strings so does
$\hstr_{x,\varepsilon}=h^{\Sym}_{x,\varepsilon}$. The shaded
\emph{symmetry gap} between $h_x$ and $\hGL_x$ is the part of the
regularity of $x$ that linear symmetry does not see; its horizontal extent
at the sufficiency line is $\soph^{\GL}(x)-\soph(x)$. It can be empty
(Corollary~\ref{cor:zerogap}) or fill the whole band
(Theorem~\ref{thm:maxgap}). The cheap curve $h^{\GL}_{x,\varepsilon}$ may
separate further from $\hGL_x$, if good linear-symmetric cells exist only
inside expensive symmetric partitions (Question~\ref{q:rigid}). The
ordering, the guides, the endpoints and the coincidences are theorems;
the shapes of the two $\GL$ curves are illustrative, and whether they obey
the halving law is open (Question~\ref{q:accept}).

\begin{figure}[t]
\centering
\begin{tikzpicture}[x=1.65cm,y=0.46cm]
\fill[violet!13] (0,10) -- (3,7) -- (3,5) -- (4.5,3.5) -- (4.5,1.5) -- (6,0)
  -- (2.5,3.5) -- (2.5,5.3) -- (0.8,7) -- (0.8,9.2) -- cycle;
\draw[->] (0,0) -- (7.3,0) node[below left]{$\alpha$};
\draw[->] (0,0) -- (0,10.8) node[above right]{$\log$-size of best model};
\draw[dashed,gray] (0,6) -- (6,0);
\draw[dashed,gray!60] (0,10) -- (7,3);
\draw[dotted,gray] (2.5,0) -- (2.5,3.5);
\draw[dotted,gray] (4.5,0) -- (4.5,1.5);
\draw[very thick,violet!80!black] (0,10) -- (3,7) -- (3,5) -- (4.5,3.5)
  -- (4.5,1.5) -- (6,0);
\draw[very thick,dashed,orange!85!black] (0,10) -- (3,7) -- (4,6) -- (4,4.5)
  -- (5.3,3.2) -- (5.3,0.7) -- (6,0);
\draw[very thick,blue!70!black] (0,10) -- (0.8,9.2) -- (0.8,7) -- (2.5,5.3)
  -- (2.5,3.5) -- (6,0) -- (7,0);
\draw[thick,densely dashed,green!55!black] (0,10) -- (0.8,9.2) -- (0.8,7)
  -- (2.5,5.3) -- (2.5,3.5) -- (6,0) -- (7,0);
\draw[line width=1.6pt,red!75!black,dash pattern=on 0.2pt off 4pt,line cap=round]
  (0,10) -- (0.8,9.2) -- (0.8,7) -- (2.5,5.3) -- (2.5,3.5) -- (6,0) -- (7,0);
\fill[blue!70!black] (2.5,3.5) circle (2.4pt);
\fill[violet!80!black] (4.5,1.5) circle (2.4pt);
\node[left] at (0,10) {$n$};
\node[left] at (0,6) {$\kappa$};
\node[below] at (0,0) {$0$};
\node[below,align=center] at (2.5,0) {$\soph(x)$\\$=\soph^{\Sym}(x)$};
\node[below] at (4.5,0) {$\soph^{\GL}(x)$};
\node[below] at (6,0) {$\kappa$};
\node at (1.9,7.0) {\small symmetry gap};
\node[gray!60!black,anchor=west] at (6.1,3.5) {\small $n-\alpha$};
\node[gray!60!black,align=left,anchor=west] at (0.1,1.6) {\small sufficiency line\\ \small $\kappa-\alpha$};
\draw[gray!60!black,->,shorten >=2pt] (0.55,2.4) -- (0.8,5.1);
\end{tikzpicture}

\vspace{4pt}
\begin{tikzpicture}[every node/.style={font=\small,anchor=west}]
\draw[very thick,blue!70!black] (0,0) -- (0.6,0); \node at (0.7,0) {$h_x$};
\draw[line width=1.6pt,red!75!black,dash pattern=on 0.2pt off 4pt,line cap=round] (1.8,0) -- (2.4,0); \node at (2.5,0) {$h^{\Sym}_x$};
\draw[thick,densely dashed,green!55!black] (3.8,0) -- (4.4,0); \node at (4.5,0) {$\hstr_{x,\varepsilon}=h^{\Sym}_{x,\varepsilon}$};
\draw[very thick,violet!80!black] (7.4,0) -- (8.0,0); \node at (8.1,0) {$\hGL_x$};
\draw[very thick,dashed,orange!85!black] (9.4,0) -- (10.0,0); \node at (10.1,0) {$h^{\GL}_{x,\varepsilon}$};
\end{tikzpicture}
\caption{The structure functions for a typical normal string. The curves
$h_x$, $h^{\Sym}_x$ and $\hstr_{x,\varepsilon}=h^{\Sym}_{x,\varepsilon}$
coincide; the first two coincide for every string. The shaded region is
the symmetry gap. The ordering, the guides, the endpoints and the
coincidences are proved; the shapes of the two $\GL$ curves are
illustrative.}
\label{fig:all}
\end{figure}

\paragraph{Halving as refinement.} Affine subspaces form an acceptable
class: a coset of dimension $d$ splits into at most $2|A|/c$ cosets of
size at most $c$. By Theorem~\ref{thm:VV}, $\haff_x$ obeys the halving
law. For $\hGL_x$ the question becomes one about the lattice $\Par_{\GL}$.
Given a symmetric $p$ with $x\in B\in p$, is there a cheap symmetric
refinement whose cell containing $x$ has size about $|B|/2$? The
algebra of Section~\ref{sec:galois} suggests meeting $p$ with a second
cheap hypothesis $q$: the cell of $x$ in $p\wedge_{\GL}q$ lies in the
intersection of its cells. Transitivity alone does not obstruct this.
$\GL(n,2)$ is $2$-transitive on nonzero vectors, so no subgroup between
$\stab(x)$ and $\GL(n,2)$ refines the orbit of $x$, but subgroups not
containing $\stab(x)$ can. Since linear maps fix $0$, homogeneous sets
other than $\{0\}$ avoid $0$, so $\A_{\GL}$ contains
$\F_2^n\setminus\{0\}$ rather than the whole space required in the
definition of acceptability; this changes nothing up to $O(1)$.

\begin{question}\label{q:accept}
Is $\A_{\GL}$ acceptable, i.e.\ can every $S\in\A_{\GL}$ of size $N$
be covered by $\mathrm{poly}(n)\,N/c$ members of $\A_{\GL}$ of size at
most $c$? In the language of Section~\ref{sec:algebra}: does every
cell admit cheap symmetric refinements, for instance by meets? If not,
does $\hGL_x$ exhibit a plateau for some $x$?
\end{question}

\begin{question}\label{q:rigid}
Is there a string $x$ for which $\hGL_x$ and $\hstr_{x,\varepsilon}$ are
both small at some level, yet $h^{\GL}_{x,\varepsilon}$ is large there?
Such an $x$ would have a good linearly homogeneous model and a good
strong model, but no cheap symmetric partition, a separation caused
by the rigidity of $\Par_{\GL}$ rather than by computability.
\end{question}

\FloatBarrier
\section{The extremes of the band}
\label{sec:extremes}

\subsection{Zero gap: symmetry explains everything}

\begin{corollary}\label{cor:zerogap}
For every admissible shape $P$ there is a string $x$ with
$h_x\approx\hGL_x\approx\haff_x\approx P$, up to
$C(P)+O(\sqrt{n\log n})$, so $\soph^{\GL}(x)=\soph(x)$ to the same
accuracy.
\end{corollary}

\begin{proof}
Theorem~\ref{thm:VV} applied to affine subspaces, and~\eqref{eq:chains}.
\end{proof}

\subsection{Maximal gap: structure invisible to symmetry}

\begin{theorem}[Maximal gap]\label{thm:maxgap}
There is a constant $c$ such that for all $n$, all $c\log n\le m\le n$
and all $t\ge0$ there are strings $x\in\{0,1\}^n$ with:
\begin{enumerate}[label=(\alph*),nosep]
\item $C(x)=m\pm O(t+\log n)$ and $h_x(\alpha)=m-\alpha\pm O(t+\log n)$ for
$c\log n\le\alpha\le m$; $x$ is stochastic and
$(O(\log n),O(t+\log n))$-normal;
\item $\hGL_x(\alpha)\ge n-\alpha-t-O(\log n)$ for all
$\alpha\le m-t-c\log n$.
\end{enumerate}
In fact all but a $2^{-t}$ fraction of a set $S$ with $C(S)=O(\log n)$
and $|S|\approx2^m$ have these properties.
\end{theorem}

\begin{proof}
\emph{Counting cells.} By Proposition~\ref{prop:GLtype}, or simply
because a group of order $N$ has a generating set of size $\log N$, there
are at most $2^{n^4}$ subgroups of $\GL(n,2)$, hence at most $2^{n^4+n}$
cells of symmetric partitions.

\emph{A set that no cell sees.} Put each string into a random $S$
independently with probability $2^{m-n}$. For a cell $B$ let
$R_B=\max(2e\,2^{m-n}|B|,n^5)$. The Chernoff bound
$\Pr[X\ge R]\le2^{-R}$ for $R\ge2e\,\mathbb{E}X$~\cite{MU} gives
$\Pr[|S\cap B|\ge R_B]\le2^{-n^5}$. A union bound, and concentration of
$|S|$, show that with positive probability
\begin{equation}\label{eq:pi}
2^{m-1}\le|S|\le2^{m+1}\quad\text{and}\quad
|S\cap B|<\max(2e\,2^{m-n}|B|,\,n^5)\ \text{for every cell }B.
\end{equation}
This property is decidable, so the lexicographically first such $S$ has
$C(S)=O(\log n)$.

\emph{Typical elements.} All but a $2^{-t}$ fraction of $x\in S$ have
$C(x)\ge m-t$; fix one.

(a) Cutting $S$ into $2^j$ consecutive pieces gives models of complexity
$j+O(\log n)$ and size at most $2^{m+1-j}$, each computable from $x$ by a
total program of length $O(\log n)$. With Lemma~\ref{lem:basic}(i) this
gives $h_x$ and normality.

(b) Fix $\alpha,s$. Fewer than $2^{\alpha+1}$ models $(p,B)$ have
$C(p,B)\le\alpha$, and the union $U$ of those cells with $|B|\le2^s$ is
enumerable from $n,m,\alpha,s$. By~\eqref{eq:pi},
$|U\cap S|\le2^{\alpha+1}\max(2e\,2^{m-n+s},n^5)$. If $x\in U$, then
\[
  m-t\le C(x)\le\log|U\cap S|+O(\log n)
  \le\alpha+\max(m-n+s,\ 5\log n)+O(\log n).
\]
For $\alpha\le m-t-c\log n$ the second branch is impossible, so
$s\ge n-\alpha-t-O(\log n)$.
\end{proof}

\begin{corollary}[Simple structure, invisible to symmetry]\label{cor:sep}
If moreover $m\le n-t-c'\log n$, these strings satisfy
$\soph(x)=O(\log n)$ but $\soph^{\GL}(x)\ge C(x)-O(t+\log n)$. They are
stochastic and normal, yet no cell of a symmetric partition costing
less than $C(x)-O(t+\log n)$ is sufficient: every sufficient symmetric
model essentially names $x$ itself. The gap between the two chains
in~\eqref{eq:chains} is a symmetry phenomenon, not a computability one.
\end{corollary}

\begin{proof}
By (b), $\hGL_x(\alpha)+\alpha\ge n-t-O(\log n)>C(x)+c\log n$ for
$\alpha\le m-t-c\log n$.
\end{proof}

\begin{remark}
The argument only uses that the partition family has
$2^{\mathrm{poly}(n)}$ cells and is decidable. So every such family,
for instance $\Par_G$ for any ambient group of order $2^{\mathrm{poly}(n)}$,
has stochastic normal strings that it sees as antistochastic. The
construction is not explicit.
\end{remark}

\section{Coordinates on the space of permutation groups}
\label{sec:coords}

The Galois connection organises symmetry hypotheses into a lattice. To
move through that lattice, and to search it, we need coordinates. This
section shows that every permutation group is a Burnside ring element
(its \emph{type}) together with a permutation (its \emph{placement}),
compares the two ambient groups in these terms, and discusses what the
maximal gap theorem implies for searching the coordinates for good
models.

\subsection{The coordinate theorem}
\label{sec:coordthm}

We give coordinates on the set of all permutation groups on a finite
set $X$, $|X|=N$; in our setting $X=\{0,1\}^n$ and $N=2^n$. Via the
Galois connection of Section~\ref{sec:galois} these become coordinates on
symmetric partitions, and they are used to navigate both spaces. To keep $G$ for the ambient
group, we write $H$ for the abstract group.

A subgroup $U\le S_X$ is the image of a faithful action
$\rho\colon H\to S_X$ of an abstract group $H\cong U$. Up to isomorphism
of $H$-sets, the action is determined by its class in the Burnside ring
$B(H)$. Fix representatives $K_1,\dots,K_r$ of the conjugacy classes of
subgroups of $H$. Then
\[
  [X]_\rho=\sum_{i=1}^r a_i\,[H/K_i],\qquad a_i\ge0,\qquad
  \deg[X]_\rho:=\sum_i a_i\,|H:K_i|=N .
\]
The action is faithful iff $\bigcap_{a_i>0}\mathrm{core}_H(K_i)=1$, where
$\mathrm{core}_H(K)=\bigcap_{h}hKh^{-1}$ is the kernel of $H$ on $H/K$.
By Burnside's theorem, $b\in B(H)$ is determined by its \emph{marks}
$\varphi_K(b)=|X_b^K|$, the number of points fixed by $K$.

\begin{definition}[Canonical realisation]
Let $B_N(H)$ be the set of faithful $b=\sum a_i[H/K_i]$ of degree $N$.
Fix $H$ concretely (with a total order on its elements) and the
representatives $K_i$ canonically. For $b\in B_N(H)$ let $X_b$ be the
disjoint union of $a_i$ copies of $H/K_i$, ordered by $i$, then by copy,
then by the least element of each coset. Identify $X_b$ with $X$ in
lexicographic order. This gives a faithful action $\rho_b$ and a
\emph{canonical subgroup} $U_b=\rho_b(H)\le S_X$.
\end{definition}

\begin{theorem}[Coordinates]\label{thm:coords}
Let $\mathcal S_H(X)=\{U\le S_X : U\cong H\}$ and
\[
  \Phi_H\colon B_N(H)\times S_X\to\mathcal S_H(X),\qquad
  \Phi_H(b,\sigma)=\sigma\,U_b\,\sigma^{-1}.
\]
\begin{enumerate}[label=(\roman*),nosep]
\item $\Phi_H$ is surjective: every $U\cong H$ has coordinates $(b,\sigma)$.
\item For fixed $b$, $\Phi_H(b,\sigma)=\Phi_H(b,\sigma')$ iff
$\sigma^{-1}\sigma'\in N_{S_X}(U_b)$.
\item $\Phi_H(b,\sigma)$ and $\Phi_H(b',\sigma')$ are conjugate in $S_X$
iff $b'\in\mathrm{Aut}(H)\cdot b$, where $\alpha\in\mathrm{Aut}(H)$ acts by
$[H/K]\mapsto[H/\alpha(K)]$.
\end{enumerate}
Consequently
\[
  \mathcal S_H(X)\;=\;\bigsqcup_{[b]\in B_N(H)/\mathrm{Aut}(H)}
  S_X/N_{S_X}(U_b):
\]
the \emph{type} $[b]$ is a discrete coordinate labelling conjugacy
classes, and the \emph{placement} $\sigma$ is a coordinate on the
homogeneous space $S_X/N_{S_X}(U_b)$. The whole space of permutation
groups on $X$ is the disjoint union of these pieces over the isomorphism
classes of $H$ with a faithful action of degree $N$.
\end{theorem}

\begin{proof}
(i) Fix an isomorphism $\psi\colon H\to U$. The $H$-set $X$ with action
$\psi$ is faithful of some class $b\in B_N(H)$, so there is an
isomorphism of $H$-sets $\sigma\colon(X,\rho_b)\to(X,\psi)$, i.e.
$\sigma\rho_b(h)\sigma^{-1}=\psi(h)$ for all $h$. Hence
$U=\sigma U_b\sigma^{-1}$. (ii) This is the definition of the
normaliser. (iii) If $\tau U_b\tau^{-1}=U_{b'}$, then
$\alpha=\rho_{b'}^{-1}\circ c_\tau\circ\rho_b$ is an automorphism of $H$,
and $\tau$ is an isomorphism of $H$-sets from $X_b$ to $X_{b'}$ twisted
by $\alpha$, so $b'=\alpha\cdot b$. The converse reverses the argument.
\end{proof}

\paragraph{Partitions and cells in coordinates.} The orbits of $U_b$ are
the summands of $X_b$, which are consecutive intervals of $X$. Let
$\lambda(b)$ be the integer partition of $N$ with part $|H:K_i|$
repeated $a_i$ times, and $p_{\lambda(b)}$ the corresponding partition of
$X$ into consecutive intervals. Then
\[
  \orb\big(\Phi_H(b,\sigma)\big)=\sigma\cdot p_{\lambda(b)} .
\]
So partitions have coordinates $(\lambda,\sigma)$, with $\sigma$ defined
modulo the stabiliser of $p_\lambda$ in $S_X$. The projection
$(b,\sigma)\mapsto(\lambda(b),\sigma)$ is $\orb$ written in coordinates:
it forgets the structure of $b$ beyond its cell sizes. If
$\sigma^{-1}(x)$ lies in a summand of type $H/K_i$, then the cell of $x$
has size $|H:K_i|$, its stabiliser is conjugate to $K_i$, and
\[
  \Pr_{u\in U}[ux=x]=\frac{|K_i|}{|H|} .
\]
The marks $\varphi_K(b)$ give the whole distribution of cell sizes and
fixed points over $X$; Proposition~\ref{prop:burnside} is the pointwise
case.

\begin{proposition}[Complexity in coordinates]\label{prop:coordcost}
Let $U=\Phi_H(b,\sigma)$, with $b$ chosen canonically in its
$\mathrm{Aut}(H)$-orbit. Then, up to $O(\log C(U))$,
\[
  C(U)\;=\;C(H,b)\;+\;\min\{\,C(\sigma'\mid H,b) : \Phi_H(b,\sigma')=U\,\}.
\]
\end{proposition}

\begin{proof}
From $U$ one computes $H$ up to isomorphism, the canonical $b$, and the
least $\sigma$ in the coset $\sigma N_{S_X}(U_b)$. This least element is
computable from any member of the coset together with $(H,b)$. Apply
symmetry of information.
\end{proof}

\paragraph{Navigation.} Two basic moves act on coordinates.
\begin{itemize}[nosep]
\item \emph{Relabelling}: $(b,\sigma)\mapsto(b,\pi\sigma)$ conjugates the
group by $\pi$ and translates its partition, at cost at most $C(\pi)$.
\item \emph{Restriction} to a subgroup $L\le H$: the subgroup
$\sigma\rho_b(L)\sigma^{-1}$ has coordinates
$(\mathrm{Res}^H_L b,\ \sigma\tau)$ in $\mathcal S_L(X)$, where $\tau$ is
computable from $(H,L,b)$. Its partition refines that of $U$ cell by
cell, according to the Mackey formula
\begin{equation}\label{eq:mackey}
  \mathrm{Res}^H_L\,[H/K]=\sum_{LhK\in L\backslash H/K}
  \big[L/(L\cap hKh^{-1})\big],
\end{equation}
so a cell of size $|H:K|$ splits into cells of sizes
$|L:L\cap hKh^{-1}|$.
\end{itemize}
Restriction moves down the subgroup lattice and refines hypotheses;
relabelling moves within a conjugacy class. In coordinates, the halving
question (Question~\ref{q:accept}) asks for a cheap subgroup $L\le H$
whose Mackey decomposition splits the cell of $x$ roughly in half.

\begin{remark}
For subgroups of $\GL(n,2)$ the same construction applies with $B(H)$
replaced by the isomorphism classes of faithful $\F_2H$-modules of
dimension $n$ (the Green ring), and $S_X$ by $\GL(n,2)$. The fibres are
then $\GL(n,2)/N_{\GL(n,2)}(U_b)$, and the placement costs at most $n^2$
bits (Proposition~\ref{prop:GLtype}).
\end{remark}

\subsection{\texorpdfstring{Types and placements: $S_X$ versus $\GL(n,2)$}{Types and placements: S\_X versus GL(n,2)}}
\label{sec:types}

The coordinates separate what a symmetry hypothesis \emph{is} (its type)
from \emph{where it sits} (its placement). The two ambient groups behave
very differently.

\begin{proposition}[Placement dominates]\label{prop:placement}
Let $S\subseteq X$, and let $p$ be any partition of $X$ having $S$ as a
cell (by Example~\ref{ex:cells}(a) there are many, the coarsest being
$\{S,X\setminus S\}$), with canonical certificate $U=\stab_{S_X}(p)$. Let $\lambda(p)$ be
the multiset of cell sizes of $p$, an integer partition of $N=2^n$.
\begin{enumerate}[label=(\roman*),nosep]
\item $U=\prod_{B\in p}\Sym(B)$ is a Young subgroup, and its type
$(H,b)$ is determined by $\lambda(p)$: $H=\prod_{B\in p}S_{|B|}$, and $b$
has one orbit of type $H/H_B$ for each cell $B$, where $H_B$ is the
stabiliser of a point of $B$. Hence
\[
  C(H,b)\le C(\lambda(p))+O(1)\le\log P(N)+O(n)=O(2^{n/2}),
\]
where $P(N)$ is the number of integer partitions of $N$.
\item Every placement realising $U$ is expensive:
\[
  \min\{C(\sigma\mid H,b) : \Phi_H(b,\sigma)=U\}\ \ge\
  C(S)-C(\lambda(p))-O(n).
\]
In particular the placement carries at least $C(S)-O(2^{n/2})$ bits,
for every partition realising $S$.
\item For $p=\{S\}\cup\{\text{singletons}\}$ with $|S|=k\ge2$, the type
is $H=S_k$, $b=[S_k/S_{k-1}]+(2^n-k)[S_k/S_k]$, costing $O(n)$, and the
placement costs $C(S)\pm O(n)$.
\end{enumerate}
\end{proposition}

\begin{proof}
(i) By Example~\ref{ex:closed}(a) the closed subgroups of $S_X$ are the
Young subgroups. A Young subgroup is determined up to conjugacy by its
cell sizes, and its abstract group and Burnside ring element are
computable from them. Given $N$, the partition $\lambda(p)$ is one of
$P(N)$ possibilities, and by the Hardy--Ramanujan formula
$\log_2P(N)\sim\pi\sqrt{2N/3}/\ln2=O(2^{n/2})$.
(ii) Since $U\le S_X$, $C(U)\le n2^n+O(1)$, so the error term in
Proposition~\ref{prop:coordcost} is $O(n)$. From $U$ one computes $p$,
and $S$ is one of at most $N$ cells of $p$. Hence
$C(S)\le C(H,b)+\min C(\sigma\mid H,b)+n+O(n)$; now use (i).
(iii) $S$ is the unique cell of size greater than one, so
$C(U)=C(S)\pm O(1)$, and the type is described by $k$ and $n$.
\end{proof}

So every symmetric realisation of a model in $S_X$ has a type that
records only cell sizes. Such a type carries at most $O(2^{n/2})$ bits,
and only in the numbers $|B|$, which is a numerical encoding rather than
a symmetry. For a model of complexity much larger than $2^{n/2}$, which
includes typical sets of non-trivial size, the placement carries
essentially all the information, under every realisation. For simpler
models a suitably chosen partition can move information into the cell
sizes, but the abstract group remains a product of full symmetric groups.
A symmetry explanation is informative only when its information lies in
a type with genuine group-theoretic structure, and the canonical
certificates of $S_X$ never provide one.

\begin{example}[The canonical certificate can hide structure]\label{ex:weight}
Let $p_{\mathrm{wt}}$ be the partition of $\{0,1\}^n$ into weight
classes. It is the orbit partition of the coordinate group $S_n$
permuting bit positions, a small group with structured type
$\sum_{w=0}^n[S_n/(S_w\times S_{n-w})]$. It is also the orbit partition
of the Young subgroup $\prod_w\Sym(\{x:|x|=w\})$, whose type is just the
list of binomial coefficients $\binom nw$. Both certify the same
hypothesis at the same canonical cost $O(\log n)$. With ambient group
$S_X$, however, the Galois connection selects the Young subgroup, so the
canonical certificate forgets the symmetry that actually explains the
partition. With ambient group $\GL(n,2)$ the closed subgroup of
$p_{\mathrm{wt}}$ is exactly the group of permutation matrices, since a
linear map preserving weight sends each $e_i$ to some $e_j$. There the
certificate keeps its structure.
\end{example}

\begin{proposition}[Type determines the hypothesis up to $n^2$ bits]\label{prop:GLtype}
Let $U\le\GL(n,2)$ have coordinates $(b,g)$ as in the remark at the end of
Section~\ref{sec:coordthm}: $b$ is the class of a faithful
$\F_2H$-module of dimension $n$ and $g\in\GL(n,2)$ is a placement. Then,
up to $O(\log n)$,
\[
  C(H,b)\ \le\ C(U)\ \le\ C(H,b)+n^2 .
\]
Consequently the number of symmetric partitions in $\Par_{\GL}$ is
at most the number of types times $|\GL(n,2)|$, which is less than
$2^{n^2}$.
\end{proposition}

\begin{proof}
The type is computable from $U$. Conversely $U=gU_bg^{-1}$, where $U_b$
is computable from the type and $g$ is an $n\times n$ matrix.
\end{proof}

Contrast Proposition~\ref{prop:placement}: in $S_X$ the placement can
carry almost all of the information, whereas in $\GL(n,2)$ it carries at
most $n^2$ bits. A linear symmetry hypothesis is essentially an abstract
group with a module, and it cannot hide arbitrary information in its
placement. Figure~\ref{fig:gl3types} shows the closed types for $n=3$.

\begin{figure}[t]
\centering
\setlength{\tabcolsep}{7pt}
\renewcommand{\arraystretch}{1.05}
\begin{tabular}{ccc}
\begin{tikzpicture}[x=1.08cm,y=1.08cm]
\draw[cubeedge] (0.00,0.00)--(1.00,0.00);
\draw[cubeedge] (0.00,0.00)--(0.00,1.00);
\draw[cubeedge] (0.00,0.00)--(0.50,0.40);
\draw[cubeedge] (1.00,0.00)--(1.00,1.00);
\draw[cubeedge] (1.00,0.00)--(1.50,0.40);
\draw[cubeedge] (0.00,1.00)--(1.00,1.00);
\draw[cubeedge] (0.00,1.00)--(0.50,1.40);
\draw[cubeedge] (1.00,1.00)--(1.50,1.40);
\draw[cubeedge] (0.50,0.40)--(1.50,0.40);
\draw[cubeedge] (0.50,0.40)--(0.50,1.40);
\draw[cubeedge] (1.50,0.40)--(1.50,1.40);
\draw[cubeedge] (0.50,1.40)--(1.50,1.40);
\filldraw[fixpt] (0.00,0.00) circle (2.6pt);
\filldraw[fixpt] (1.00,0.00) circle (2.6pt);
\filldraw[fixpt] (0.00,1.00) circle (2.6pt);
\filldraw[fixpt] (1.00,1.00) circle (2.6pt);
\filldraw[fixpt] (0.50,0.40) circle (2.6pt);
\filldraw[fixpt] (1.50,0.40) circle (2.6pt);
\filldraw[fixpt] (0.50,1.40) circle (2.6pt);
\filldraw[fixpt] (1.50,1.40) circle (2.6pt);
\end{tikzpicture}
&
\begin{tikzpicture}[x=1.08cm,y=1.08cm]
\draw[cubeedge] (0.00,0.00)--(1.00,0.00);
\draw[cubeedge] (0.00,0.00)--(0.00,1.00);
\draw[cubeedge] (0.00,0.00)--(0.50,0.40);
\draw[cubeedge] (1.00,0.00)--(1.00,1.00);
\draw[cubeedge] (1.00,0.00)--(1.50,0.40);
\draw[cubeedge] (0.00,1.00)--(1.00,1.00);
\draw[cubeedge] (0.00,1.00)--(0.50,1.40);
\draw[cubeedge] (1.00,1.00)--(1.50,1.40);
\draw[cubeedge] (0.50,0.40)--(1.50,0.40);
\draw[cubeedge] (0.50,0.40)--(0.50,1.40);
\draw[cubeedge] (1.50,0.40)--(1.50,1.40);
\draw[cubeedge] (0.50,1.40)--(1.50,1.40);
\filldraw[fixpt] (0.00,0.00) circle (2.6pt);
\fill[wA] (1.00,0.00) circle (3pt);
\fill[wB] (0.00,1.00) circle (3pt);
\filldraw[fixpt] (1.00,1.00) circle (2.6pt);
\filldraw[fixpt] (0.50,0.40) circle (2.6pt);
\fill[wA] (1.50,0.40) circle (3pt);
\fill[wB] (0.50,1.40) circle (3pt);
\filldraw[fixpt] (1.50,1.40) circle (2.6pt);
\end{tikzpicture}
&
\begin{tikzpicture}[x=1.08cm,y=1.08cm]
\draw[cubeedge] (0.00,0.00)--(1.00,0.00);
\draw[cubeedge] (0.00,0.00)--(0.00,1.00);
\draw[cubeedge] (0.00,0.00)--(0.50,0.40);
\draw[cubeedge] (1.00,0.00)--(1.00,1.00);
\draw[cubeedge] (1.00,0.00)--(1.50,0.40);
\draw[cubeedge] (0.00,1.00)--(1.00,1.00);
\draw[cubeedge] (0.00,1.00)--(0.50,1.40);
\draw[cubeedge] (1.00,1.00)--(1.50,1.40);
\draw[cubeedge] (0.50,0.40)--(1.50,0.40);
\draw[cubeedge] (0.50,0.40)--(0.50,1.40);
\draw[cubeedge] (1.50,0.40)--(1.50,1.40);
\draw[cubeedge] (0.50,1.40)--(1.50,1.40);
\filldraw[fixpt] (0.00,0.00) circle (2.6pt);
\fill[wA] (1.00,0.00) circle (3pt);
\filldraw[fixpt] (0.00,1.00) circle (2.6pt);
\fill[wA] (1.00,1.00) circle (3pt);
\filldraw[fixpt] (0.50,0.40) circle (2.6pt);
\fill[wA] (1.50,0.40) circle (3pt);
\filldraw[fixpt] (0.50,1.40) circle (2.6pt);
\fill[wA] (1.50,1.40) circle (3pt);
\end{tikzpicture}
\\
\small $1$, order $1$, $1$ placement & \small $C_2$, order $2$, $21$ placements & \small $V_4^{\mathrm{ax}}$, order $4$, $7$ placements\\
\footnotesize\phantom{x} & \footnotesize\phantom{x} & \footnotesize\phantom{x}\\[8pt]
\begin{tikzpicture}[x=1.08cm,y=1.08cm]
\draw[cubeedge] (0.00,0.00)--(1.00,0.00);
\draw[cubeedge] (0.00,0.00)--(0.00,1.00);
\draw[cubeedge] (0.00,0.00)--(0.50,0.40);
\draw[cubeedge] (1.00,0.00)--(1.00,1.00);
\draw[cubeedge] (1.00,0.00)--(1.50,0.40);
\draw[cubeedge] (0.00,1.00)--(1.00,1.00);
\draw[cubeedge] (0.00,1.00)--(0.50,1.40);
\draw[cubeedge] (1.00,1.00)--(1.50,1.40);
\draw[cubeedge] (0.50,0.40)--(1.50,0.40);
\draw[cubeedge] (0.50,0.40)--(0.50,1.40);
\draw[cubeedge] (1.50,0.40)--(1.50,1.40);
\draw[cubeedge] (0.50,1.40)--(1.50,1.40);
\filldraw[fixpt] (0.00,0.00) circle (2.6pt);
\fill[wA] (1.00,0.00) circle (3pt);
\fill[wB] (0.00,1.00) circle (3pt);
\fill[wC] (1.00,1.00) circle (3pt);
\filldraw[fixpt] (0.50,0.40) circle (2.6pt);
\fill[wA] (1.50,0.40) circle (3pt);
\fill[wB] (0.50,1.40) circle (3pt);
\fill[wC] (1.50,1.40) circle (3pt);
\end{tikzpicture}
&
\begin{tikzpicture}[x=1.08cm,y=1.08cm]
\draw[cubeedge] (0.00,0.00)--(1.00,0.00);
\draw[cubeedge] (0.00,0.00)--(0.00,1.00);
\draw[cubeedge] (0.00,0.00)--(0.50,0.40);
\draw[cubeedge] (1.00,0.00)--(1.00,1.00);
\draw[cubeedge] (1.00,0.00)--(1.50,0.40);
\draw[cubeedge] (0.00,1.00)--(1.00,1.00);
\draw[cubeedge] (0.00,1.00)--(0.50,1.40);
\draw[cubeedge] (1.00,1.00)--(1.50,1.40);
\draw[cubeedge] (0.50,0.40)--(1.50,0.40);
\draw[cubeedge] (0.50,0.40)--(0.50,1.40);
\draw[cubeedge] (1.50,0.40)--(1.50,1.40);
\draw[cubeedge] (0.50,1.40)--(1.50,1.40);
\filldraw[fixpt] (0.00,0.00) circle (2.6pt);
\fill[wA] (1.00,0.00) circle (3pt);
\fill[wA] (0.00,1.00) circle (3pt);
\fill[wA] (1.00,1.00) circle (3pt);
\fill[wB] (0.50,0.40) circle (3pt);
\filldraw[fixpt] (1.50,0.40) circle (2.6pt);
\fill[wB] (0.50,1.40) circle (3pt);
\fill[wB] (1.50,1.40) circle (3pt);
\end{tikzpicture}
&
\begin{tikzpicture}[x=1.08cm,y=1.08cm]
\draw[cubeedge] (0.00,0.00)--(1.00,0.00);
\draw[cubeedge] (0.00,0.00)--(0.00,1.00);
\draw[cubeedge] (0.00,0.00)--(0.50,0.40);
\draw[cubeedge] (1.00,0.00)--(1.00,1.00);
\draw[cubeedge] (1.00,0.00)--(1.50,0.40);
\draw[cubeedge] (0.00,1.00)--(1.00,1.00);
\draw[cubeedge] (0.00,1.00)--(0.50,1.40);
\draw[cubeedge] (1.00,1.00)--(1.50,1.40);
\draw[cubeedge] (0.50,0.40)--(1.50,0.40);
\draw[cubeedge] (0.50,0.40)--(0.50,1.40);
\draw[cubeedge] (1.50,0.40)--(1.50,1.40);
\draw[cubeedge] (0.50,1.40)--(1.50,1.40);
\filldraw[fixpt] (0.00,0.00) circle (2.6pt);
\fill[wA] (1.00,0.00) circle (3pt);
\fill[wB] (0.00,1.00) circle (3pt);
\fill[wA] (1.00,1.00) circle (3pt);
\filldraw[fixpt] (0.50,0.40) circle (2.6pt);
\fill[wA] (1.50,0.40) circle (3pt);
\fill[wB] (0.50,1.40) circle (3pt);
\fill[wA] (1.50,1.40) circle (3pt);
\end{tikzpicture}
\\
\small $V_4^{\mathrm{ce}}$, order $4$, $7$ placements & \small $S_3$, order $6$, $28$ placements & \small $D_8$, order $8$, $21$ placements\\
\footnotesize\phantom{x} & \footnotesize closure of $C_3$ & \footnotesize closure of $C_4$\\[8pt]
\begin{tikzpicture}[x=1.08cm,y=1.08cm]
\draw[cubeedge] (0.00,0.00)--(1.00,0.00);
\draw[cubeedge] (0.00,0.00)--(0.00,1.00);
\draw[cubeedge] (0.00,0.00)--(0.50,0.40);
\draw[cubeedge] (1.00,0.00)--(1.00,1.00);
\draw[cubeedge] (1.00,0.00)--(1.50,0.40);
\draw[cubeedge] (0.00,1.00)--(1.00,1.00);
\draw[cubeedge] (0.00,1.00)--(0.50,1.40);
\draw[cubeedge] (1.00,1.00)--(1.50,1.40);
\draw[cubeedge] (0.50,0.40)--(1.50,0.40);
\draw[cubeedge] (0.50,0.40)--(0.50,1.40);
\draw[cubeedge] (1.50,0.40)--(1.50,1.40);
\draw[cubeedge] (0.50,1.40)--(1.50,1.40);
\filldraw[fixpt] (0.00,0.00) circle (2.6pt);
\fill[wA] (1.00,0.00) circle (3pt);
\filldraw[fixpt] (0.00,1.00) circle (2.6pt);
\fill[wA] (1.00,1.00) circle (3pt);
\fill[wA] (0.50,0.40) circle (3pt);
\fill[wA] (1.50,0.40) circle (3pt);
\fill[wA] (0.50,1.40) circle (3pt);
\fill[wA] (1.50,1.40) circle (3pt);
\end{tikzpicture}
&
\begin{tikzpicture}[x=1.08cm,y=1.08cm]
\draw[cubeedge] (0.00,0.00)--(1.00,0.00);
\draw[cubeedge] (0.00,0.00)--(0.00,1.00);
\draw[cubeedge] (0.00,0.00)--(0.50,0.40);
\draw[cubeedge] (1.00,0.00)--(1.00,1.00);
\draw[cubeedge] (1.00,0.00)--(1.50,0.40);
\draw[cubeedge] (0.00,1.00)--(1.00,1.00);
\draw[cubeedge] (0.00,1.00)--(0.50,1.40);
\draw[cubeedge] (1.00,1.00)--(1.50,1.40);
\draw[cubeedge] (0.50,0.40)--(1.50,0.40);
\draw[cubeedge] (0.50,0.40)--(0.50,1.40);
\draw[cubeedge] (1.50,0.40)--(1.50,1.40);
\draw[cubeedge] (0.50,1.40)--(1.50,1.40);
\filldraw[fixpt] (0.00,0.00) circle (2.6pt);
\fill[wA] (1.00,0.00) circle (3pt);
\fill[wB] (0.00,1.00) circle (3pt);
\fill[wA] (1.00,1.00) circle (3pt);
\fill[wB] (0.50,0.40) circle (3pt);
\fill[wA] (1.50,0.40) circle (3pt);
\fill[wB] (0.50,1.40) circle (3pt);
\fill[wA] (1.50,1.40) circle (3pt);
\end{tikzpicture}
&
\begin{tikzpicture}[x=1.08cm,y=1.08cm]
\draw[cubeedge] (0.00,0.00)--(1.00,0.00);
\draw[cubeedge] (0.00,0.00)--(0.00,1.00);
\draw[cubeedge] (0.00,0.00)--(0.50,0.40);
\draw[cubeedge] (1.00,0.00)--(1.00,1.00);
\draw[cubeedge] (1.00,0.00)--(1.50,0.40);
\draw[cubeedge] (0.00,1.00)--(1.00,1.00);
\draw[cubeedge] (0.00,1.00)--(0.50,1.40);
\draw[cubeedge] (1.00,1.00)--(1.50,1.40);
\draw[cubeedge] (0.50,0.40)--(1.50,0.40);
\draw[cubeedge] (0.50,0.40)--(0.50,1.40);
\draw[cubeedge] (1.50,0.40)--(1.50,1.40);
\draw[cubeedge] (0.50,1.40)--(1.50,1.40);
\filldraw[fixpt] (0.00,0.00) circle (2.6pt);
\fill[wA] (1.00,0.00) circle (3pt);
\fill[wA] (0.00,1.00) circle (3pt);
\fill[wA] (1.00,1.00) circle (3pt);
\fill[wA] (0.50,0.40) circle (3pt);
\fill[wA] (1.50,0.40) circle (3pt);
\fill[wA] (0.50,1.40) circle (3pt);
\fill[wA] (1.50,1.40) circle (3pt);
\end{tikzpicture}
\\
\small $S_4^{\mathrm{ln}}$, order $24$, $7$ placements & \small $S_4^{\mathrm{pl}}$, order $24$, $7$ placements & \small $\GL_3$, order $168$, $1$ placement\\
\footnotesize closure of $A_4^{\mathrm{ln}}$ & \footnotesize closure of $A_4^{\mathrm{pl}}$ & \footnotesize closure of $C_7,\ 7{:}3$\\[8pt]
\end{tabular}
\caption{The nine Galois-closed types of subgroups of $\GL(3,2)$, one placement each, with the $3$-cube coloured by orbits (hollow corners are fixed points). A type with normaliser $N$ has $168/|N|$ placements, and each closed placement gives one symmetric partition, $100$ in all. The six non-closed types are listed under their closures: they give no new partitions. Superscripts distinguish the classes stabilising a line (ln) or a plane (pl), and the Klein four-groups with a common axis (ax) or centre (ce). An interactive version showing every placement is in the HTML edition.}
\label{fig:gl3types}
\end{figure}

\subsection{What the maximal gap theorem means for search}
\label{sec:search}

In coordinates a hypothesis is a type $(H,b)$ with a placement, and one
moves by relabelling, by restriction with its Mackey splitting of cells,
and by meets. Theorem~\ref{thm:maxgap} limits every search of this kind.
For its strings no linear-symmetric cell much cheaper than $C(x)$ is a
good model, so no search can find one; and since these strings are
stochastic and normal, ``no good symmetric model found'' means ``no
linear-symmetric structure found'', not that $x$ is random or strange.
Conversely, a model found by search gives only an upper bound on
$\hGL_x$, never a matching lower bound.

The two ambient groups fail in opposite ways. For $S_X$ search is
complete in principle (Theorem~\ref{thm:sym}), but the canonical
certificates are Young subgroups, so for complex models the search takes
place almost entirely among placements (Proposition~\ref{prop:placement}):
it is a search over arbitrary finite sets under another name. For
$\GL(n,2)$ the placement costs at most $n^2$ bits
(Proposition~\ref{prop:GLtype}), so there are only $2^{\mathrm{poly}(n)}$
hypotheses. This makes search feasible, and it is exactly what the proof
of Theorem~\ref{thm:maxgap} exploits.

\FloatBarrier
\section{Algorithms: searching for orbit models}
\label{sec:algorithms}

This section makes the search of Section~\ref{sec:search} precise and
gives its algorithms. Given a string $x$, an
agent moves through the lattice of subgroups of an ambient group $G$,
down to maximal subgroups or sideways by cheap relabellings, looking for
a subgroup under which $x$ has a short description. The lattices are far
too large to enumerate, so, as in AlphaZero~\cite{Silver18}, a learned
value function evaluates positions that cannot be searched to the end and
a tree search spends computation where that function is uncertain.
Subgroups are a good search space because every orbit of a cheap group is
a strong model (Proposition~\ref{prop:basicP}(ii)), because the lattice
has known local structure, so that every move is a computable step with a
code length, and because the codelength of $x$ is read off coordinates
without listing the $2^n$ points of $X$. The section reports no
experiments.

\paragraph{Conventions.}
$X=\{0,1\}^n$ and $N=2^n$. In
Sections~\ref{sec:alg-objective}--\ref{sec:alg-exact} the letter $H$
denotes a subgroup of the ambient group $G\le\Sym(X)$, the hypothesis
written $U$ earlier; in Section~\ref{sec:alg-sym} $H$ is again the
abstract group of Section~\ref{sec:coords}. We write $\norb(H)$ for the
number of orbits of $H$ on $X$, $H\cdot x$ for the orbit of $x$,
$\stab_H(x)$ for its stabiliser, $\Fix(u)=\{y:uy=y\}$, and $\Fix(H)$ for
the points fixed by all of $H$. The potential $\Phi_\lambda$ of
Definition~\ref{def:potential} should not be confused with the coordinate
map $\Phi_H(b,\sigma)$, which always has two arguments.

\subsection{The objective and its collapse}\label{sec:alg-objective}

\begin{definition}[Codelength under a hypothesis]\label{def:codelength}
The \emph{codelength of $x$ under the coding scheme $H$} is
$\CL{H}{x}=\log\norb(H)+\log|H\cdot x|$: the index of the orbit of $x$,
then its position in the orbit.
\end{definition}

In the three-part description~\eqref{eq:threepart}, $\CL{H}{x}$ is the
term $C(B\mid p)+\log|B|$ with the cell named by a uniform index, and
ceilings dropped~\cite{Grunwald07}. It depends on $H$ only through
$\orb(H)$, so it is unchanged under closure
(Proposition~\ref{prop:galois}).

\begin{definition}[Excess]
For $U\le\Sym(X)$ put
$\varepsilon_U=\frac1N\sum_{u\in U\setminus\{1\}}|\Fix(u)|$ and
$\delta_U=\log(1+\varepsilon_U)$. For $U\le\GL(n,2)$,
$|\Fix(u)|=2^{\,n-\rank(u-I)}$.
\end{definition}

\begin{proposition}[Score identity]\label{prop:identity}
For every $H\le\Sym(X)$ and $x\in X$,
\[
\norb(H)=\frac N{|H|}\,(1+\varepsilon_H),\qquad
\CL{H}{x}=n+\delta_H-\log|\Stab_H(x)|.
\]
Moreover $\delta$ is monotone: $U'\le U$ implies $\delta_{U'}\le\delta_U$.
\end{proposition}

\begin{proof}
By Burnside's lemma, $\norb(H)=\frac1{|H|}\sum_{u\in H}|\Fix(u)|$; the
identity contributes $N$, the rest $N\varepsilon_H$. By orbit--stabiliser,
$|H\cdot x|=|H|/|\Stab_H(x)|$; adding the logarithms, $\log|H|$ cancels.
$\varepsilon_U$ is a sum of non-negative terms over $U\setminus\{1\}$.
\end{proof}

If $p_H$ chooses an orbit of $H$ uniformly and then a point of it
uniformly, $\CL{H}{x}=-\log p_H(x)$: minimising $\mathcal L$ is maximum
likelihood, which overfits without a model cost; adding the cost of $H$
gives two-part minimum description length~\cite{Rissanen78,Grunwald07}.
The overfitting is extreme.

\begin{theorem}[Collapse]\label{thm:collapse}
For every $H\le\Sym(X)$ and $x\in X$,
$\min_{K\le H}\CL{K}{x}=\CL{\Stab_H(x)}{x}=\log\norb\big(\Stab_H(x)\big)$,
and $\CL{H}{x}-\min_{K\le H}\CL{K}{x}=\delta_H-\delta_{\Stab_H(x)}$.
\end{theorem}

\begin{proof}
Write $S_K=\Stab_K(x)$. By Proposition~\ref{prop:identity},
$\CL{K}{x}=n+\delta_K-\log|S_K|$ and
$\CL{S_K}{x}=n+\delta_{S_K}-\log|S_K|$, so
$\CL{S_K}{x}\le\CL{K}{x}$ by monotonicity of $\delta$. A group $L$ fixing
$x$ has $\CL{L}{x}=\log\norb(L)$, and $L\le L'$ implies
$\norb(L)\ge\norb(L')$. For $K\le H$, $S_K\le S_H$, so
$\CL{K}{x}\ge\CL{S_K}{x}=\log\norb(S_K)\ge\log\norb(S_H)=\CL{S_H}{x}$.
The last claim is the difference of the two formulas with $K=H$.
\end{proof}

\begin{corollary}[Top values]\label{cor:top}
$\min_{K\le G}\CL{K}{x}=\log\norb(G_x)$ with $G_x=\Stab_G(x)$. This is $1$
bit for the $2$-transitive groups $\Sym(X)$ and $\AGL(n,2)$, and
$\log3\approx1.58$ bits for $\GL(n,2)$, $n\ge2$, $x\ne0$, whose point
stabiliser has orbits $\{0\}$, $\{x\}$ and $\F_2^n\setminus\{0,x\}$.
\end{corollary}

The minimiser names $x$: for these groups $\Fix(G_x)$ is $\{x\}$ or
$\{0,x\}$, so $C(G_x)=C(x)\pm O(\log n)$. An uncosted score therefore
pulls every policy towards stabilisers of $x$, which explain nothing, and
the uncosted search needs no learning, since its value is known in closed
form. All learnable structure lies in the trade-off between $\mathcal L$
and the cost of the hypothesis.

\subsection{The cost of a hypothesis and the strong sublattice}\label{sec:alg-cost}

\begin{definition}[Description scheme]\label{def:scheme}
A \emph{description scheme} is a computable partial map that, given $n$,
sends words $w$ of a prefix-free set $\mathcal W_n\subseteq\{0,1\}^*$ to
generating sets of subgroups $H_w\le G$. A state that records its word
has cost $|w|$.
\end{definition}

\begin{proposition}[Sufficiency floor]\label{prop:floor}
$C(H_w\mid n)\le|w|+O(1)$, and $C(x)\le|w|+\CL{H_w}{x}+O(\log n)$. The
point $(\alpha,\beta)=(|w|+\lceil\log\norb(H_w)\rceil,\ \log|H_w\cdot x|)$
satisfies $h^G_x(\alpha+O(1))\le\beta$, and if $|w|\le\varepsilon$ also
$h^G_{x,\varepsilon+O(1)}(\alpha+O(1))\le\beta$.
\end{proposition}

\begin{proof}
Describe $x$ by $w$, the index of its orbit and its index in the orbit;
$(\orb(H_w),H_w\cdot x)$ is a $\Par_G$-model with $C(p)\le|w|+O(1)$.
\end{proof}

So the search operates on descriptions, not bare groups, and every model
it reports is a certified upper bound on the profile. Explicit generators
are never compressive (a permutation of $X$ costs about $n2^n$ bits, a
matrix $n^2$), so a scheme must name groups by short paths through the
lattice, by a structured language of named constructions on the cube, or
by an $x$-blind learned prior $q$ used as an arithmetic code
(\emph{amortised MDL}).

\begin{lemma}[Proposal--code separation]\label{lem:separation}
If the search chooses moves by any procedure that may read $x$, but
charges each move its length under a code depending only on the current
word and $n$, then Proposition~\ref{prop:floor} holds at every state. If
the charge may depend on $x$, it can fail.
\end{lemma}

\begin{proof}
The decoder reads moves one at a time with codes it can compute. For the
converse, charge $1$ bit for the move sequence leading to $G_x$ when the
data is $x$: then $1+\CL{G_x}{x}<C(x)$ for most $x$.
\end{proof}

\begin{definition}[Potential]\label{def:potential}
For $\lambda\ge0$, $\Phi_\lambda(w,x)=\lambda|w|+\CL{H_w}{x}$. For
$\lambda=1$ this is the total two-part codelength.
\end{definition}

The cost enters as a Lagrangian (minimise $\Phi_\lambda$), as a budget
(minimise $\mathcal L$ subject to $|w|\le\varepsilon$, the search version
of $h^G_{x,\varepsilon}$), or both. The budget enforces the restriction to
strong models, but the cheap subgroups are only computably enumerable and
form neither an ideal, a filter nor a union of conjugacy classes. The
following decidable inner approximation is what a budgeted search
explores.

\begin{definition}[Chain entropy]\label{def:chainentropy}
Fix, computably from $n$, a finite set $\mathcal S$ of \emph{seeds} with
prefix-free code lengths $c(R)$, and an alphabet $\Gamma\subseteq G$ of
conjugating letters with code lengths $c(\gamma)$. The \emph{move graph}
has an edge $H\to M$ for each maximal subgroup $M$ of $H$, of weight
$1+\log r(H)+\log|H:N_H(M)|$, where $r(H)$ is the number of conjugacy
classes of maximal subgroups, and an edge $H\to\gamma H\gamma^{-1}$ of
weight $1+c(\gamma)$. The \emph{chain entropy} $d_{\mathcal S}(H)$ is the
least $c(R)$ plus path weight over seeds $R$ and paths from $R$ to $H$,
and $B_\varepsilon=\{H:d_{\mathcal S}(H)\le\varepsilon\}$ is the
\emph{strong sublattice}.
\end{definition}

\begin{theorem}[Soundness of the strong sublattice]\label{thm:chainentropy}
$C(H\mid n)\le d_{\mathcal S}(H)+O(\log n)$. Hence $B_\varepsilon$ is
decidable, $|B_\varepsilon|\le(\varepsilon+1)2^\varepsilon$, and every
cell of every $H\in B_\varepsilon$ is an $(\varepsilon+O(\log n))$-strong
model.
\end{theorem}

\begin{proof}
At each node the weights are the lengths of a two-stage code (class, then
conjugate) with a flag bit, so they define a sub-probability over
outgoing edges. Arithmetic-code the path with these sub-probabilities. At
most $2^\varepsilon$ paths of each length have weight at most
$\varepsilon$, and every edge weighs at least $1$.
\end{proof}

The root alone is not enough: for $\Sym(X)$, $N\ge5$, the only subgroup
of index below $N$ is $\Alt(X)$, and for $\GL(n,2)$, $n\ge5$, the smallest
index is $2^n-1$~\cite{KL90}, so every first step from the root costs
about $n$ bits. The scheme needs structured seeds and named moves
(Sections~\ref{sec:alg-sym} and~\ref{sec:alg-gl}); their choice decides
which symmetric structure counts as simple.

\subsection{The lattice as a decision process}\label{sec:alg-mdp}

We use the standard notions of reinforcement
learning~\cite{SB18,BT96}. Our process is deterministic and episodic, and
we take the discount $\gamma=1$: discounting an improvement found deeper
in the lattice has no meaning for code lengths. Rewards are differences
of a potential, so returns telescope~\cite{NHR99}.

\begin{definition}[Lattice decision process]\label{def:mdp}
Fix a description scheme, a budget $\varepsilon$ and a weight
$\lambda\ge0$. A \emph{state} is $s=(w,x,\lambda,b)$ with group $H=H_w$
and remaining budget $b=\varepsilon-|w|$. The \emph{actions} are
$\mv{Down}(M)$ for a maximal subgroup $M$ of $H$, $\mv{Side}(\gamma)$ for
a cheap letter $\gamma\in\Gamma$ (coordinate transpositions, bit flips,
transvections), each with a code $\Delta w$ and feasible if
$|\Delta w|\le b$, and \mv{Stop}. Transitions are $w\mapsto w\Delta w$.
The \emph{reward} is
$r(s,a,s')=\Phi_\lambda(s)-\Phi_\lambda(s')=\CL{H}{x}-\CL{H'}{x}-\lambda|\Delta w|$,
and \mv{Stop} earns $0$.
\end{definition}

The return from $s_t$ is $\Phi_\lambda(s_t)-\Phi_\lambda(s_T)$, the
weighted bits removed before stopping. We write $\mathcal L_t$ or
$\mathcal L_v$ for the codelength at a state or search node. Every episode
has at most $\varepsilon$ moves, since every code has at least one bit.
Each state is a point $(|w|+\log\norb(H_w),\log|H_w\cdot x|)$ of the
profile of $x$ (Figure~\ref{fig:sf}), and the level sets of $\Phi_1$ are
the lines of slope $-1$. Episodes start at a seed or a charged random
word; a random conjugate by a uniform placement is not admissible, since
the placement alone costs about $n^2$ bits in $\GL(n,2)$ and up to
$n2^n$ in $\Sym(X)$. Up moves $\langle H,\pi\rangle$ coarsen partitions
but create cycles, and are used only for exploration.

\subsection{Value functions on the lattice}\label{sec:alg-values}

Let $\R(s)$ be the set of states reachable from $s$, $s$ included.

\begin{proposition}[Value as best reachable model]\label{prop:bellman}
With $m(s)=\min_{s'\in\R(s)}\Phi_\lambda(s')$,
\[
V^*(s)=\Phi_\lambda(s)-m(s)\ge0,\qquad Q^*(s,a)=\Phi_\lambda(s)-m(s_a),
\]
and $m(s)=\min\{\Phi_\lambda(s),\min_a m(s_a)\}$ over feasible moves. $V^*$
does not depend on the cost $|w|$ already spent.
\end{proposition}

\begin{proof}
The Bellman equation with $\gamma=1$, deterministic transitions and
\mv{Stop} worth $0$ is $V^*(s)=\max\{0,\max_a[\Phi(s)-\Phi(s_a)+V^*(s_a)]\}$;
substitute $V^*=\Phi-m$. The term $\lambda|w|$ cancels.
\end{proof}

So the value of a subgroup is the gap between its own model and the best
model below it, and the best model in the class is read off the value at
the root: \emph{model search is value estimation}.

\begin{proposition}[Anatomy of a down move]\label{prop:anatomy}
For $M\le H$, with $S_H=\Stab_H(x)$ and $S_M=M\cap S_H$,
$\CL{H}{x}-\CL{M}{x}=(\delta_H-\delta_M)-\log|S_H:S_M|$: excess removed
minus stabiliser lost. The subgroups $M\in[S_H,H]$ correspond to the
blocks of imprimitivity of $H$ on $H\cdot x$ containing
$x$~\cite[Thm.~1.5A]{DM96}; for them $\Stab_M(x)=S_H$, the cell of $x$
shrinks to the block $M\cdot x$, and $\mathcal L$ does not increase.
\end{proposition}

A cheap block of half the orbit is the halving question
(Question~\ref{q:accept}) in lattice form, and a search that stalls in an
interval $[S_H,H]$ without cheap blocks may have reached a genuine plateau
of $h^G_x$.

\begin{theorem}[Value sandwich]\label{thm:sandwich}
For down moves, every $\lambda\ge0$ and every budget,
\[
0\le V^*(H,x,\lambda,b)\le\overline V(H,x):=\delta_H-\delta_{\Stab_H(x)},
\]
and every state reachable from $s$ has potential at least
$\Phi_{\mathrm{low}}(s)=\lambda|w|+\log\norb(\Stab_H(x))$. Along down
moves $\overline V$ does not increase and $\Phi_{\mathrm{low}}$ does not
decrease.
\end{theorem}

\begin{proof}
\mv{Stop} gives $V^*\ge0$. Reachable groups lie below $H$ and words only
grow, so Theorem~\ref{thm:collapse} gives both bounds and the
monotonicity.
\end{proof}

Both bounds are computable from $H$ and $x$. They are used for pruning, for
a \emph{bounded value head} that predicts a fraction
$\rho_\theta\in[0,1]$ and outputs $\widehat V_\theta=\rho_\theta\overline V$,
so that no prediction violates the sandwich, and as the exact value for
$\lambda=0$: the learned part of the value is the effect of the cost.

\begin{proposition}[Cost removes the spurious value]\label{prop:costvalue}
Let $\lambda=1$ and $\defi(s)=\Phi_1(s)-C(x)$. Then
$V^*(s)\le\defi(s)+O(\log n)$. At the root, the uncosted value is about
$n$ bits for every $x$ ($n-1$ for $\Sym(X)$ and $\AGL(n,2)$,
$\log\tfrac{2(N-1)}3$ for $\GL(n,2)$), whereas
$V^*_{\lambda=1}(G)\le n-C(x)+O(\log n)$.
\end{proposition}

\begin{proof}
By Proposition~\ref{prop:floor} every reachable state has
$\Phi_1\ge C(x)-O(\log n)$; the uncosted value is
Theorem~\ref{thm:collapse} with Corollary~\ref{cor:top}.
\end{proof}

So a large costed value certifies, up to $O(\log n)$, that $x$ has short
symmetric descriptions. Planted strings give the converse in a controlled
form: if $P=H_{w_P}$ is reachable and $y$ is uniform on $\Fix(P)$, then
with probability at least $1-2^{-c}$,
$\defi(w_P)\le|w_P|+\log\norb(P)-\log|\Fix(P)|+c+O(\log n)$, since $y$ is
fixed and typical in $\Fix(P)$. They supply validation targets with a
known near-sufficient model. Conjugation by $g\in G$ preserves
$\mathcal L$ and $\overline V$, and preserves $V^*$ for those $g$ that
leave code lengths unchanged; so the cost-free part of the value depends
only on the conjugacy class of $H$ and the position of $x$ relative to
it, and the cost part is carried by the word $w$.

\subsection{Evaluation, certified targets and training}\label{sec:alg-train}

\begin{algorithm}[ht]
\caption{\mv{Evaluate}$(w,x,\lambda)$}\label{alg:evaluate}
\begin{algorithmic}[1]
\Require word $w$; string $x\ne0$; weight $\lambda$
\State $H\gets H_w$;\quad $\norb(H)\gets\frac1{|H|}\sum_{C\in\mathrm{ccl}(H)}|C|\,|\Fix(c_C)|$
  \Comment{Burnside over classes}
\State $\delta_H\gets\log(\norb(H)|H|/N)$;\quad $S\gets\Stab_H(x)$;\quad $\delta_S$ likewise
\State $\mathcal L\gets n+\delta_H-\log|S|$;\quad $\Phi\gets\lambda|w|+\mathcal L$
\State $\overline V\gets\delta_H-\delta_S$;\quad $\Phi_{\mathrm{low}}\gets\lambda|w|+\log\norb(S)$
  \Comment{Theorem~\ref{thm:sandwich}}
\State \Return $\mathcal L,\Phi,\overline V,\Phi_{\mathrm{low}}$ and invariant features of $(H,x)$
\end{algorithmic}
\end{algorithm}

No orbit on $X$ is enumerated: the cost is that of conjugacy classes and
one point stabiliser in $H$. Candidate moves are generated from $x$, since
random maximal subgroups almost never relate to it: interval moves to
stabilisers of blocks containing $x$, conjugates of maximal subgroups
that keep most of $\Stab_H(x)$, ambient-specific constructions, cheap
relabellings and a few random moves. Each is charged its $x$-independent
code length (Lemma~\ref{lem:separation}), and a policy $\pi_\theta$,
which may see $x$, ranks them; it is used for search, never for coding.

\paragraph{Certified targets.}
Let $D$ be a prefix-closed set of feasible words collected by any
procedure and for any strings, and for $w\in D$ let $D_w$ be its
extensions in $D$. Put
$V_D(w,x)=\Phi_\lambda(w,x)-\min_{u\in D_w}\Phi_\lambda(u,x)$.

\begin{lemma}[Certified targets]\label{lem:targets}
For every $D$, string $x$ and weight $\lambda$,
$0\le V_D(w,x)\le V^*(w,x,\lambda,\varepsilon-|w|)$, and
$V_D=\Phi_\lambda-m_D$ with $m_D(w)=\min\{\Phi_\lambda(w),\min_c m_D(c)\}$
over the children $c$ of $w$ in $D$.
\end{lemma}

\begin{proof}
Each $u\in D_w$ is reached from $w$ by recorded feasible moves whose codes
depend neither on the string nor on the policy, so
$\Phi_\lambda(w,x)-\Phi_\lambda(u,x)$ is the return of a feasible policy
for $x$ and at most $V^*$ by Proposition~\ref{prop:bellman}.
\end{proof}

For an episode with potentials $\Phi_t$, the target is the gain to the
best later stopping point, $L_t=\Phi_t-\min_{t'\ge t}\Phi_{t'}$, computed
in one backward pass. It is a lower bound on $V^*(s_t)$ for data from
\emph{any} policy, including the random one of the first round, and stays
valid under replay. The on-policy return is not used: it credits $s_t$
with where the episode happened to stop and can be negative. Because the
lattice, the moves and their costs do not depend on $x$, the lemma holds
for every string, not only the one that guided the exploration: one
explored graph gives certified targets for a whole batch of strings by
one bottom-up pass each, the analogue of hindsight experience
replay~\cite{Andrychowicz17}. The bound is tight only for strings whose
good models lie in $D$, so for strings other than the guide only
predictions outside $[V_D,\overline V]$ are penalised.

\paragraph{Training.}
The loop is generalised policy iteration~\cite[\S4.6]{SB18}. Phase~A
regresses $\widehat V_\theta=\rho_\theta\overline V$ on the targets $L_t$
of episodes with exploring starts~\cite[\S5.3]{SB18}, beginning with
random down-chains. Phase~B runs the lattice search of
Section~\ref{sec:alg-mcls} with the current network and trains on its
certified values and, for $\pi_\theta$, on the visit distribution of
children, as in AlphaZero~\cite{Silver18}. An optional Phase~C fits an
$x$-blind prior $q$ to the move sequences of the best descriptions found
and freezes it as the code, so that hypotheses useful across many strings
become cheaper. Training data are natural strings encoded so that their
symmetries act on coordinates (image patches, cellular automata, truth
tables, codewords), planted strings for validation, and uniform strings
as a control on which values should be near zero. Episodes store
$\Phi_t$ at every state, so that targets can be recomputed when $\lambda$
or the cost scheme changes.

\subsection{Monte Carlo lattice search}\label{sec:alg-mcls}

\emph{Monte Carlo lattice search} (MCLS) is Monte Carlo tree
search~\cite{KS06,Browne12} adapted to the lattice. It searches a DAG of
subgroups with one node per subgroup, shared across paths
(transpositions~\cite{CBK08,SCM12}); it backs up minima rather than means,
because the problem is single-agent and deterministic and the quantity of
interest is the best model found~\cite{Schadd08}; it widens progressively
from $x$-guided proposals~\cite{Couetoux11}; and it prunes with the
lattice bound of Theorem~\ref{thm:sandwich}.

Since a node may be reached by paths of different costs, it stores not a
potential but the \emph{local optimum}
$\mu_v=\min_{v'\text{ below }v}[\lambda\,\mathrm{cost}(v\to v')+\mathcal L_{v'}]$,
which satisfies $V^*(v)=\mathcal L_v-\mu^*_v$ and the recursion
$\mu_v=\min\{\mathcal L_v,\min_c(\lambda e_{vc}+\mu_c)\}$ on any DAG,
with $e_{vc}$ the code length of the edge. MCLS keeps a \emph{realised}
$\mu_v$ from explored nodes and an \emph{estimated} $\tilde\mu_v$ in which
unexplored leaves contribute $\mathcal L_\ell-\widehat V_\theta(\ell)$.

\begin{algorithm}[ht]
\caption{\mv{MCLS}$(x,\lambda,\varepsilon,K_{\mathrm{sim}})$}\label{alg:mcls}
\begin{algorithmic}[1]
\Require string $x$; weight $\lambda$; budget $\varepsilon$; number of simulations $K_{\mathrm{sim}}$
\State create root $\rho$: $\mu_\rho\gets \mathcal L_\rho$, $\tilde\mu_\rho\gets \mathcal L_\rho-\widehat V_\theta(\rho)$, $N_\rho\gets0$; nodes keyed by subgroup
\For{$k=1,\dots,K_{\mathrm{sim}}$}
  \State $v\gets\rho$;\quad $g\gets0$;\quad path $\gets[\rho]$ \Comment{$g$: cost spent from the root}
  \Loop \Comment{selection}
    \If{$|\mathrm{children}(v)|<\lceil c_wN_v^{\,\alpha}\rceil$} \textbf{break} \Comment{progressive widening}
    \EndIf
    \State $c\gets\argmax_{c}\Big[\dfrac{\mathcal L_v-\lambda e_{vc}-\tilde\mu_c}{\max(\overline V_v,1)}
       +c_{\mathrm{puct}}\,\pi_\theta(c\mid v)\dfrac{\sqrt{N_v}}{1+N_c}\Big]$ over unpruned children
    \State $g\gets g+\lambda e_{vc}$;\quad $v\gets c$;\quad append $v$ to path
  \EndLoop
  \State $c\gets$ node of the next untried proposal at $v$; record $e_{vc}$ \Comment{expansion}
  \If{$g+\lambda e_{vc}+\log\norb(\Stab_{H_c}(x))\ge\mu_\rho$} mark edge pruned;\ \textbf{continue}
  \EndIf
  \If{$c$ is new} evaluate $c$ (Algorithm~\ref{alg:evaluate}); $\mu_c\gets \mathcal L_c$;
     $\tilde\mu_c\gets \mathcal L_c-\widehat V_\theta(c)$ \EndIf
  \For{$u$ on the path, from $v$ back to $\rho$} \Comment{min-backup}
    \State $\mu_u\gets\min\big(\mu_u,\ \min_{c'}(\lambda e_{uc'}+\mu_{c'})\big)$;\quad
           $\tilde\mu_u\gets\min\big(\tilde\mu_u,\ \min_{c'}(\lambda e_{uc'}+\tilde\mu_{c'})\big)$;\quad
           $N_u\gets N_u+1$
  \EndFor
\EndFor
\State \Return the archive of all nodes with their profile points, and the node realising $\mu_\rho$
\end{algorithmic}
\end{algorithm}

\begin{proposition}[Soundness of MCLS]\label{prop:mclssound}
At every stage, $\mathcal L_v-\mu_v\le V^*(v)$ for every node $v$;
$\mu_\rho$ is the potential of an actual hypothesis; and a pruned edge
cannot lead to a hypothesis with potential below $\mu_\rho$.
\end{proposition}

\begin{proof}
$\mu_v$ starts at $\mathcal L_v$ and is lowered only to values realised by
explored edges. Through a pruned edge the potential is at least
$g+\lambda e_{vc}$ plus the bound of Theorem~\ref{thm:sandwich}, hence at
least $\mu_\rho$; the subtree remains reachable through other edges,
which are tested separately.
\end{proof}

The exploration bonus keeps revisiting branches whose estimates are
uncertain, so an early underestimate by $\widehat V_\theta$ is corrected by
what the search finds. With a hard budget, a cheaper path to an existing
node replaces the old one.

\subsection{Exact values on small lattices}\label{sec:alg-exact}

When the whole lattice below a root is available, optimal values are
computed exactly by dynamic programming from the bottom up. This is the
ground truth for validation: on $S_4$ ($30$ subgroups), on $\GL(3,2)$
($179$ subgroups; Figure~\ref{fig:gl3types} shows its closed types),
on $\GL(4,2)\cong A_8$, and below any Burnside hypothesis with a small
abstract group (Proposition~\ref{prop:sublattice}).

\begin{algorithm}[ht]
\caption{\mv{ExactValue}$(x,\lambda)$}\label{alg:exact}
\begin{algorithmic}[1]
\Require all subgroups below a root, with their maximal subgroups and edge code lengths $e_{HM}$
\For{subgroups $H$ in increasing order of $|H|$}
  \State $\mu^*_H\gets\min\big\{\CL{H}{x},\ \min_{M\ \text{maximal in}\ H}(\lambda e_{HM}+\mu^*_M)\big\}$
\EndFor
\State \Return $V^*(H)=\CL{H}{x}-\mu^*_H$ for all $H$
\end{algorithmic}
\end{algorithm}

Validation compares the learned and exact values (error, rank
correlation, and violations of the sandwich, which the bounded head
excludes), the search regret $\mu_\rho-\mu^*_{\mathrm{root}}$ as a
function of simulations, and the recovered profile against the exact
one. For $\lambda=0$ the identity $V^*=\delta_H-\delta_{\Stab_H(x)}$ of
Theorem~\ref{thm:collapse} tests the implementation of $\mathcal L$.

\subsection{All permutations: Burnside coordinates}\label{sec:alg-sym}

With ambient group $\Sym(X)$ the search is complete in principle but
placement dominates (Proposition~\ref{prop:placement}), and closures
forget the structure of a group (Example~\ref{ex:weight}). The
permutation search therefore works with descriptions in the coordinates
of Section~\ref{sec:coords}, never with closures. Let $H$ be a small
abstract group with class representatives $K_1,\dots,K_r$ of its
subgroups, $b=\sum_ja_j[H/K_j]$ a faithful type of degree $N$, $|b|=\sum_ja_j$,
and $U=\Phi_H(b,\sigma)=\sigma\rho_b(H)\sigma^{-1}$. The point
$\sigma^{-1}x$ lies in a summand of type $j(x)$ at a coset $h_0K_{j(x)}$,
with stabiliser $K_x=h_0K_{j(x)}h_0^{-1}$.

\begin{theorem}[$\mathcal L$ in Burnside coordinates]\label{thm:mackeyJ}
$\CL{U}{x}=\log|b|+\log|H:K_{j(x)}|$. For $L\le H$ and
$U_L=\sigma\rho_b(L)\sigma^{-1}$, with $S=L\cap K_x$,
\[
\CL{U_L}{x}=\log\sum_ja_j\,|L\backslash H/K_j|+\log|L:S|,\qquad
\overline V(U_L,x)=\log\frac{\sum_ja_j|L\backslash H/K_j|\cdot|L:S|}{\sum_ja_j|S\backslash H/K_j|}.
\]
\end{theorem}

\begin{proof}
The orbits of $U$ are the summands. By the Mackey
formula~\eqref{eq:mackey} each summand $H/K_j$ splits under $L$ into
$|L\backslash H/K_j|$ orbits, and the $L$-orbit of the coset of
$\sigma^{-1}x$ has size $|L:L\cap K_x|$. The bound is
Theorem~\ref{thm:collapse} applied to $U_L$, whose minimiser is $U_S$.
\end{proof}

\begin{proposition}[The lattice below a hypothesis]\label{prop:sublattice}
$L\mapsto U_L$ is a lattice isomorphism from $\mathrm{Sub}(H)$ onto the
subgroups of $U$. So the down-set of a Burnside hypothesis is the subgroup
lattice of its abstract group, independently of $N$, $b$ and $\sigma$,
and every value of $\mathcal L$ on it is read off the coefficients $a_j$,
the position of $\sigma^{-1}x$ and the table of double-coset counts of $H$.
\end{proposition}

This is the main computational fact of the permutation search. The
lattice below a hypothesis has $|\mathrm{Sub}(H)|$ elements ($30$ for
$S_4$, $156$ for $S_5$, $1455$ for $S_6$), does not grow with $n$, and is
fully tabulated once per abstract group, with its double-coset counts and
table of marks~\cite{Pfeiffer97}. Evaluation applies the letters of the
placement word to the single point $x$, locates it in the layout and reads
tables; no permutation of the $2^n$ points is ever built. A state
$(H,b,L,\sigma)$ costs an index of $H$ in a library, a sparse code for
$b$, at most $\log|\mathrm{Sub}(H)|$ bits for $L$, and the length of the
placement word. Moves are restriction to maximal subgroups of $L$ (cells
split by Mackey), relabelling by cheap letters, type moves that add,
remove or swap summands, and extension to an overgroup in the library.
Burnside coordinates thus enter the lattice at a cheap, structured point
instead of paying about $n$ bits for the first step from the root.
Natural seeds are the structured actions on the cube, each described in
$O(\log n)$ bits: $\AGL(n,2)$ and the translations, $\GL(n,2)$, the
coordinate permutations $S_n$, for which
$\CL{S_n}{x}=\log(n+1)+\log\binom n{|x|}$ is the familiar two-part code by
Hamming weight, block products $S_k\times S_{n-k}$, and translations
supported on a set of coordinates.

\subsection{Linear symmetry: \texorpdfstring{$\GL(n,2)$}{GL(n,2)}}\label{sec:alg-gl}

For $G=\GL(n,2)$ and $x\ne0$, $|\Fix(u)|=2^{\,n-\rank(u-I)}$, so
$\norb(H)$ and $|H\cdot x|=|H:\Stab_H(x)|$ come from linear algebra over
$\F_2$; $\Fix(H)$ is the null space of the stacked matrix
$[u_1-I;\dots;u_g-I]$ over generators, and point stabilisers are computed
through a permutation representation (Schreier--Sims). Closures keep
structure here, and a placement costs at most $n^2$ bits. Maximal
subgroups are classified by Aschbacher's theorem~\cite{Asch84,KL90}, but a
$k$-dimensional subspace costs about $k(n-k)$ bits under the
class-then-conjugate code, so the search starts from seeds: block-diagonal
groups on coordinate intervals, pointwise stabilisers $1_W\times\GL(W')$
of coordinate subspaces, permutation matrices, the unitriangular group,
Singer cycles and their subgroups, and $\Sp(n,2)$ for the standard form.
It moves by restriction and by transvections and coordinate
transpositions of $O(\log n)$ bits each, and proposes stabilisers of
subspaces built from $x$, such as $\langle x\rangle$ and
$\Fix(\Stab_H(x))$.

By Corollary~\ref{cor:zerogap} the search has targets; by
Theorem~\ref{thm:maxgap} a small root value means ``no linear structure
found'', not ``$x$ is random''. Since $\GL(n,2)\le\AGL(n,2)\le\Sym(X)$ and
$\mathcal L$ depends only on orbits, the linear search is the down-set of
one seed of the permutation search, and every profile point it finds is
available to the permutation search at an additive $O(1)$ cost. Running
both is informative: if only the
permutation search finds a near-sufficient cheap model, $x$ has simple
structure that is not linear-symmetric, and the cost split of that model
shows whether it is a symmetry with an informative type or a named set.
\FloatBarrier

\end{document}